\documentclass[twocolumn]{revtex4-2}
\usepackage{iceberg}
\usepackage{nicefrac}
\usepackage{amsmath}
\usepackage{thmtools}

\usepackage{tikz}
\usepackage{xcolor}
\usepackage{graphicx}

\usepackage{booktabs}
\usepackage{multirow}

\usepackage[caption=false]{subfig}

\definecolor{mygreen}{RGB}{40,90,180}
\definecolor{myblue}{RGB}{225, 85, 35}
\definecolor{myaux}{RGB}{20,140,60}

\makeatletter

\renewenvironment{algorithm}[1][]{%
  \par\addvspace{\bigskipamount}\noindent
  \minipage{\linewidth}%
  \hrule height.8pt depth0pt \kern4pt%
  \refstepcounter{algorithm}%
  \def\caption##1{\par\noindent\textbf{Algorithm~\thealgorithm:}~##1\par\kern4pt\hrule\kern4pt}%
}{%
  \kern4pt\hrule height.8pt\relax
  \endminipage\par\addvspace{\bigskipamount}%
}
\makeatother
\crefname{algorithm}{algorithm}{algorithms}
\Crefname{algorithm}{Algorithm}{Algorithms}

\notestrue

\DeclareMathOperator{\cone}{\text{cone}}

\DeclareMathOperator{\Hom}{\text{Hom}}
\DeclareMathOperator{\kk}{\mathbb{K}}

\newcommand{\mm}{\mathfrak{m}}

\DeclareMathOperator{\ls}{\mathrm{LS}}

\makeatletter
\AddToHook{cmd/appendix/before}{\def\cref@section@alias{appendix}\def\cref@subsection@alias{appendix}}
\makeatother

\begin{document}
\title{Lifted surgery: Fast processing with QLDPC codes}
\author{Lucas Berent}
\affiliation{Iceberg Quantum}
\author{Lawrence Z.\ Cohen}
\affiliation{Iceberg Quantum}
\author{Armanda O.\ Quintavalle}
\affiliation{Iceberg Quantum}

\author{{\scriptsize\texttt{\{lucas,larry,armanda\}@iceberg-quantum.com}}}

\begin{abstract}
Quantum low-density parity-check (QLDPC) codes are a leading candidate for achieving low-overhead fault-tolerant quantum computing.
However, the time overhead of logical operations in QLDPC codes remains a key challenge.
Code surgery, a space-efficient technique for fault-tolerant logical measurements, incurs this overhead through repeated rounds of syndrome measurement.
We introduce lifted surgery, a method for fast and parallel surgery on Abelian group algebra codes that maintains the low physical overhead that makes QLDPC codes attractive.
Lifted surgery preserves the symmetries of the underlying code, making searches for large instances tractable and offering a natural route towards efficient hardware implementations.
We utilise block decompositions and techniques from commutative algebra to characterise lifted surgery, investigate well-behaved subfamilies, and construct explicit examples.
In particular, we present quantum radial codes with parameters $\llbracket 90, 8, 10 \rrbracket$ and $\llbracket 198, 8, 16 \rrbracket$ for which surgery is fast, parallel, and addressable, allowing arbitrary sets of independent logical operators of the same Pauli type to be measured in a single round of syndrome extraction.
We benchmark lifted surgery under circuit-level depolarising noise and, for the $\llbracket 90, 8, 10 \rrbracket$ code, find logical performance comparable to standard code surgery while requiring ten times fewer rounds of syndrome measurement.
By combining speed and parallelism, lifted surgery offers a practical route towards low-overhead fault-tolerant quantum computing.
\end{abstract}

\maketitle

\section{Introduction}
Fault-tolerant architectures will be essential for utility-scale quantum computing.
By using quantum error correction to redundantly encode logical information, they allow us to arbitrarily suppress errors and thereby unlock useful quantum algorithms.
This protection, however, comes with substantial space-time overheads.
For example, architectures based on the surface code would require around $1$ million physical qubits to break RSA-$2048$~\cite{gidney_how_2025}.
Quantum low-density parity-check (QLDPC) codes have emerged as promising candidates for reducing this overhead~\cite{breuckmann_quantum_2021,webster_explicit_2025,bravyi_highthreshold_2024,bhardwaj_highrate_2026}.
By exploiting long-range connectivity, QLDPC codes achieve high encoding rates while retaining good error correction capabilities.
On platforms that can realise the required form of connectivity, QLDPC-based architectures could therefore enable scalable, low-overhead fault-tolerant quantum computation~\cite{webster_pinnacle_2026,yoder_tour_2025,cain_shor_2026}.
For example, the Pinnacle architecture only requires $100,000$ physical qubits to break RSA-$2048$.

Storing logical information coherently is not sufficient for performing fault-tolerant algorithms; we also need methods for performing computation on the encoded information.
A common approach for QLDPC codes is to compile a quantum algorithm from a universal gate set into a sequence of logical Pauli measurements~\cite{litinski_game_2019,webster_pinnacle_2026,he_extractors_2025}.
These measurements can be performed using code surgery~\cite{cohen_lowoverhead_2022,ide_faulttolerant_2025,williamson_lowoverhead_2026,he_extractors_2025}, which deforms the error-correcting code using an auxiliary system such that syndrome measurements on the deformed code return the value of the desired logical operator.

This has a significant time overhead, since noisy syndrome measurements require repeating the procedure a number of times proportional to the distance of the code to ensure fault tolerance.
\emph{Fast surgery} removes this scaling, enabling constant-time fault-tolerant logical measurements by using an auxiliary system with additional properties that protect against syndrome measurement errors~\cite{baspin_fast_2025,cowtan_fast_2025}.
However, although straightforward constructions reduce the time cost, they incur a large qubit overhead and hence exhibit the same space-time overhead as standard code surgery schemes.
Practical constructions of auxiliary systems with the required properties remain unknown in general, apart from techniques specialised to hypergraph product codes~\cite{chang_constanttime_2026}.
In this work, we bridge this gap by giving explicit constructions of practical auxiliary systems for fast surgery on general Abelian group algebra codes, a widespread family that underpins several proposed QLDPC architectures~\cite{webster_pinnacle_2026,yoder_tour_2025,cain_shor_2026,tripier_faulttolerant_2026}.

\subsection{Overview of main results}
\begin{figure*}[th]
    \includegraphics[width=0.95\textwidth]{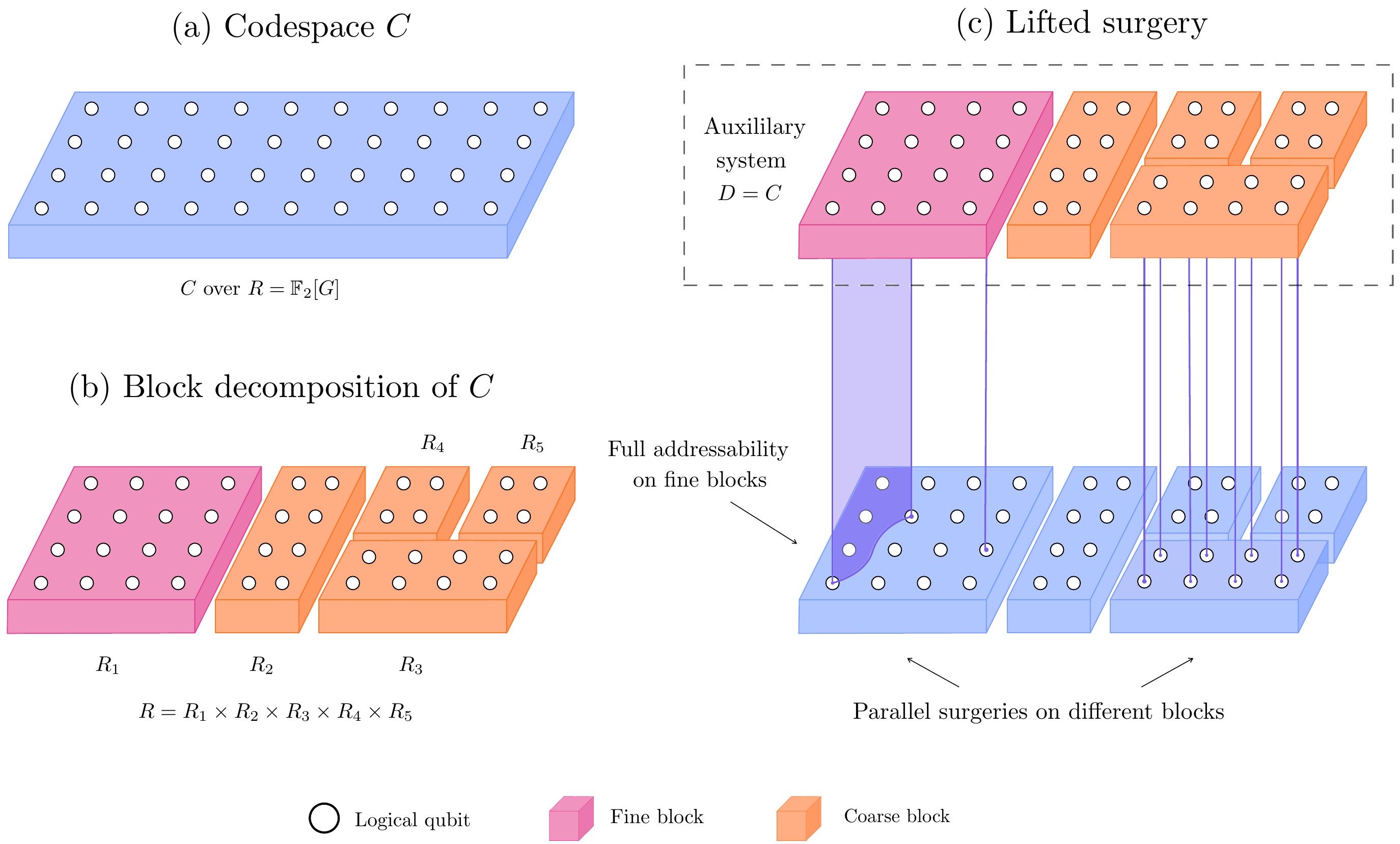}
    \caption{
        Graphical representation of lifted surgery. 
        (a) Codespace of $C_\bullet$ over the group algebra $R$ (blue).
        (b) The codespace decomposes into the same blocks as the group algebra $R$. 
        (c) The top row (magenta and orange) shows the auxiliary complex, a copy of $C_\bullet$.
        Different lifted surgeries can be combined to act in parallel on different blocks. 
        When the blocks are fields, as in the quantum radial code example, lifted surgery allows for single logical qubit addressability.
    }
    \label{fig:four-diagrams}
\end{figure*}
Our main idea is to exploit the algebraic symmetries already present in group algebra codes to construct practical fast surgery schemes.
We consider quantum codes obtained as lifts of chain complexes over $\Ft$-group algebras and begin with a complete description of their block structure.
The group algebra canonically decomposes into a product of local rings, and the codespace and every linear map over the group algebra decompose accordingly into independent blocks. 
This decomposition exposes the structure and symmetries of the logical subspace and provides the natural framework in which we define \emph{lifted surgery}. 
These central ideas are illustrated in~\Cref{fig:four-diagrams}. 

To construct fast surgery schemes given a group algebra code, we take an identical copy of the code as an auxiliary system and study the linear maps from the auxiliary system to the base code that yield ``good'' surgery schemes.
Crucially, the block structure turns this construction problem into a tractable computation.
Lifted surgery operations organise into equivalence classes according to the logical measurement they implement, allowing us to efficiently determine the possible lifted surgery operations for a given code and then search within each class for representatives with favourable QEC parameters, such as low check weight.

We then investigate structured subfamilies of lifted surgery operations.
We first introduce \emph{scalar surgery}, the simplest such family. 
For codes described by a Koszul complex, scalar surgery yields a merged code that is again described by a (higher-dimensional) Koszul complex.
This closure property is valuable in practice since the merged code inherits favourable properties, such as low-depth syndrome extraction circuits, from the base code.
Moreover, for codes that have a product structure, we introduce \emph{lifted product surgery}, which allows us to construct particularly sparse merged codes.

For quantum radial codes~\cite{scruby_highthreshold_2026}, lifted surgery achieves fast, parallel, and fully addressable logical measurements.
For these codes, we can construct lifted surgery maps for any subgroup of the logical Pauli group consisting entirely of $Z$-type (or entirely of $X$-type) operators.
For a $\llbracket 90, 8, 10 \rrbracket$ radial code, we optimise lifted surgery maps for all rank-1 and rank-2 $Z$-type measurements and in every case find distance-preserving maps with maximum check weight and qubit degree at most 11.
For a $\llbracket 198, 8, 16 \rrbracket$ radial code, we optimise all rank-1 $Z$-type measurements and find distance-preserving maps with maximum check weight and qubit degree at most 14.
We also construct lifted surgery maps for several instances of two-block group algebra codes.
These examples demonstrate that the systolic expansion conditions imposed on auxiliary systems in Refs.~\cite{baspin_fast_2025, cowtan_fast_2025} are sufficient but not necessary for ensuring distance preservation.

Furthermore, to measure logical operators across separate codes, we generalise the fast bridging procedure introduced in Ref.~\cite{cowtan_fast_2025}.
Unlike this prior construction, for addressable families of codes such as the radial codes, our bridging construction allows us to jointly measure nonequivalent logical Pauli operators across copies of the same code using lifted surgery.
For the $\llbracket 90, 8, 10 \rrbracket$ and $\llbracket 198, 8, 16 \rrbracket$ radial codes, we give explicit instances of distance-preserving gadgets that bridge any pair of logical operators of the same type.

Lastly, we construct syndrome extraction circuits for lifted surgery on several group algebra codes and use integer linear programming~\cite{vittal_flagproxy_2024} to schedule the two-qubit gates in the circuits.
By exploiting the underlying group algebra symmetry, we reduce the size of this optimisation problem and make the circuit search tractable.
We also develop a practical decoding pipeline for the resulting, potentially dense, detector matrices by combining detector sparsification with correlated decoding.
We simulated lifted surgery under circuit-level noise.
For the $\llbracket 90,8,10 \rrbracket$ code, we find that lifted surgery performs comparably to standard code surgery while requiring an order of magnitude fewer rounds of syndrome measurement.
We additionally simulated bridged cross code-block lifted surgery and multiple lifted surgery operations performed sequentially with no intervening memory padding.
We find that the logical performance per operation matches that of isolated operations.

\subsection{Organisation of the manuscript}
In~\Cref{sec:prelim}, we give an overview of code surgery and briefly review the state of the art.
Then~\Cref{sec:ls-theory} discusses the block decomposition for chain complexes over $\Ft$ group algebras and introduces lifted surgery formally.
In~\Cref{sec:ls-constructions}, we give explicit lifted surgery constructions as well as detailed examples for several code instances.
In~\Cref{sec:bridging}, we introduce bridged lifted surgery for joint measurements across codes.
Finally, we present results of numerical decoding simulations in~\Cref{sec:numerics} and conclude with a discussion and outlook in~\Cref{sec:discussion}.

\section{Fault-Tolerant logical measurements}\label{sec:prelim}

In this section we give a general overview of code surgery for QLDPC codes, including the recent developments that aim to reduce the time overhead for performing surgery.

The core primitives in code surgery are measurements of logical Pauli product operators.
Surgery thus gives a method to realise Pauli-based computation~\cite{litinski_game_2019}.
Pauli product measurements (PPMs) are sufficient for implementing the entire Clifford group.
Supplementing PPMs with magic states thus achieves universal quantum computation.
Consequently, code surgery along with a supply of magic states is sufficient to execute any quantum algorithm.
Originally, surgery was introduced for the planar surface code, where lattice surgery is used for multi-qubit operations~\cite{bombin_quantum_2009,horsman_surface_2012}.
For QLDPC codes with higher rates and non-local connectivity, Pauli measurements are also often more natural, more efficient, and easier to implement than other techniques, such as transversal operations.
Note that since we can implement magic state injection from auxiliary logical qubits as a sequence of Pauli measurements~\cite{litinski_game_2019}, Pauli-based computation is universal. 
The technique for implementing Pauli measurements on general QLDPC codes in a fault-tolerant manner is code surgery, originally introduced in Ref.~\cite{cohen_lowoverhead_2022}.

\subsection{Code surgery}
As mentioned above, the goal of code surgery is to non-destructively measure the value of a logical operator in a code block in a fault-tolerant manner.
The most naive way to achieve this would be to simply promote the logical operator to a stabiliser generator and measure this generator using a single check qubit.
The problem with this approach is that this single check qubit needs to be coupled to every data qubit in the support of the logical operator.
In fault-tolerant quantum computing, however, we want to use codes with sufficiently high distances in order to exponentially suppress logical errors.
Furthermore, it is often very useful to perform PPMs of products of many disjoint Pauli operators.
To meet these two needs, the above method would require coupling the check qubit to potentially \emph{hundreds} of physical qubits, which would inevitably result in a scheme that is not fault-tolerant.

To circumvent this problem, Ref.~\cite{cohen_lowoverhead_2022} introduced a surgery scheme for performing PPMs on arbitrary QLDPC codes while maintaining the LDPC property and the distance of the underlying code.
This scheme generalises lattice surgery~\cite{bombin_quantum_2009,horsman_surface_2012}, which is tailored to performing PPMs on local two-dimensional surface codes.
Instead of measuring the stabiliser generator corresponding to the logical operator with a single check qubit, we decompose it into multiple generators by introducing a set of auxiliary qubits that together form an auxiliary system -- an auxiliary code.
Ref.~\cite{cohen_lowoverhead_2022} achieved this by taking a mirror copy of the logical operator to be measured and then thickening it (by taking a tensor product with a repetition code) at most $d$ times in order to ensure the distance is preserved.

While the above technique gives a flexible method for performing PPMs, one downside is that the thickening procedure can introduce a significant spatial overhead.
This problem was solved in Refs.~\cite{williamson_lowoverhead_2026,ide_faulttolerant_2025}, which gave general formalisms for performing distance-preserving surgery by using an expander graph as the auxiliary system.
Intuitively, this expansion in the auxiliary system ensures that the logical operator support in the base code cannot be cleaned into the auxiliary system by multiplication with stabilisers in a weight-reducing manner.

Subsequent work has focused on improving these results in several different directions. 
Refs.~\cite{cross_improved_2025, swaroop_universal_2024} introduced bridging, which gives a systematic method for connecting separate auxiliary systems.
Several works have explored how to connect multiple auxiliary systems to a code simultaneously, allowing for the measurement of sets of commuting logical operators in parallel~\cite{zhang_timeefficient_2025, cowtan_parallel_2026}.
Similarly, several works have explored using hypergraphs, instead of standard graphs, to perform high-rate surgery~\cite{zheng_highrate_2025, zheng_logical_2026} on particular code families.
Refs.~\cite{he_extractors_2025, blue_full_2026} consider designing surgery schemes tailored to hardware platforms with fixed connectivity.

One unifying theme across much of the recent research into code surgery is the use of chain complexes and homology.
A quantum CSS code can be described by the following diagram:
\begin{equation*}
    \begin{tikzcd}
    	Z & Q & X
    	\arrow["{H_Z^T}", from=1-1, to=1-2]
    	\arrow["{H_X}", from=1-2, to=1-3]
    \end{tikzcd}
\end{equation*}
Here, $Z$, $Q$, and $X$ are based $\Ft$-vector spaces that represent the $Z$-checks, data qubits, and $X$-checks.
The maps $H_Z^T$ and $H_X$ are the parity-check matrices, and we require that $H_X H_Z^T = 0$, which corresponds to the commutativity of the $X$- and $Z$-checks.

Now suppose that we want to measure a $Z$-type logical operator in the above code (here and throughout the paper, analogous results hold for $X$ measurements by swapping labels).
The auxiliary system is similarly described by the diagram:
\begin{equation*}
    \begin{tikzcd}
    	V & E & F
    	\arrow["{\partial_2}", from=1-1, to=1-2]
    	\arrow["{\partial_1}", from=1-2, to=1-3]
    \end{tikzcd}
\end{equation*}
It is common to associate this diagram with a (hyper)graph, wherein $V$ represents the vertices, $E$ represents the (hyper)edges, and $F$ represents the graph cycles (or elements of $\ker \partial_2$ in the hypergraph case).
In the merged code, elements of $V$ will correspond to $Z$-type checks, whose products are used to infer the logical measurement outcome; 
elements of $E$ will correspond to data qubits, and elements of $F$ will correspond to $X$-type checks.

We describe the coupling between the auxiliary system and the code block by two ``port'' functions, $f_1: V \to Q$ and $f_0: E \to X$.
Here, $f_1$ describes which data qubits the auxiliary $Z$-type checks act on, and $f_0$ describes how the code block's $X$-type checks are deformed by auxiliary data qubits.
This is described by the following diagram:
\begin{equation*}
    \begin{tikzcd}
    	V & E & F \\
    	Z & Q & X
    	\arrow["{\partial_2}", from=1-1, to=1-2]
    	\arrow["{f_1}", from=1-1, to=2-2]
    	\arrow["{\partial_1}", from=1-2, to=1-3]
    	\arrow["{f_0}", from=1-2, to=2-3]
    	\arrow["{H_Z^T}"', from=2-1, to=2-2]
    	\arrow["{H_X}"', from=2-2, to=2-3]
    \end{tikzcd}
\end{equation*}
For this to constitute a valid merged code, that is, one where the checks commute, we must satisfy the condition $H_X f_1 = f_0 \partial_2$ (the diagram commutes).
Furthermore, to ensure that the merged code is sparse, we want all the port functions, as well as the maps in the auxiliary system, to be sparse.
If we want to measure a specific subspace, $L$, of $Z$-type logical operators, then we require that the image of $f_1$ is $L$.

The surgery scheme then proceeds straightforwardly:
We start by initialising the data qubits, $E$, in the $\ket{+}$ state.
We then perform $O(d)$ rounds of syndrome extraction with the new checks, $Z \cup V$ ($Z$-type checks) and $X \cup F$ ($X$-type checks), where $d$ is the distance of the base code.
Lastly, we measure out the data qubits, $E$, in the $X$ basis.
We can then infer the logical measurement outcomes by taking products of outcomes of the $V$ checks (up to Pauli corrections).

The $O(d)$ rounds of syndrome extraction are necessary to ensure fault tolerance of the procedure in the presence of noisy syndrome measurements.
If only one round of syndrome extraction is performed, then a single bit flip on one of the outcomes in $V$ can flip the outcome of the logical measurement.
For utility-scale quantum computing, we often require codes with distances $\geq 20$, and hence $O(d)$ rounds of syndrome extraction can represent a significant time overhead.
In particular, for systems with slow clock cycles, such as trapped ions and neutral atoms, this can result in runtimes of decades unless significant space-time trade-offs are made~\cite{webster_pinnacle_2026,zhou_resource_2025}.

\subsection{Fast surgery}
Fast surgery, introduced in Refs.~\cite{baspin_fast_2025,cowtan_fast_2025}, is a generalisation of the formalism discussed above, intended to address the significant time overhead of code surgery.
The auxiliary system in standard surgery can be viewed as a classical code, where $V$ represents the bits, $E$ represents the classical checks, and $\partial_2$ is the check matrix.
In fast surgery, we instead use a quantum code as an auxiliary system.
This setup can be represented by the following diagram:
\begin{equation*}
    \begin{tikzcd}
    	M & V & E & F
    	\arrow["{\partial_3}", from=1-1, to=1-2]
    	\arrow["{\partial_2}", from=1-2, to=1-3]
    	\arrow["{\partial_1}", from=1-3, to=1-4]
    \end{tikzcd}
\end{equation*}
The reason for using a quantum code is that in the merged code the elements of $M$ will become \emph{meta-checks} that capture redundancy amongst the $Z$-type checks.
This means that $\partial_3$ is a check matrix for the $Z$ syndrome measurements, and by including it in the decoding matrix we can correct errors in the $Z$ syndromes, even when only a constant number of rounds of syndrome extraction have been performed.
As in the standard surgery case, we can describe the deformed code with the following diagram:
\begin{equation}\label{eq:fast-surgery}
    \begin{tikzcd}
    	M & V & E & F \\
    	0 & Z & Q & X
    	\arrow["{\partial_3}", from=1-1, to=1-2]
    	\arrow["{f_2}"', from=1-1, to=2-2]
    	\arrow["{\partial_2}", from=1-2, to=1-3]
    	\arrow["{f_1}"', from=1-2, to=2-3]
    	\arrow["{\partial_1}", from=1-3, to=1-4]
    	\arrow["{f_0}"', from=1-3, to=2-4]
    	\arrow["0"', from=2-1, to=2-2]
    	\arrow["{H_Z^T}"', from=2-2, to=2-3]
    	\arrow["{H_X}"', from=2-3, to=2-4]
    \end{tikzcd}
\end{equation}
We now need to add a third port function, $f_2: M \to Z$, along with the commutativity condition $H_Z^T f_2 = f_1 \partial_3$.
For the remainder of the manuscript, we will drop the $F$ term, since, given choices for $\partial_2$ and $E$, $\partial_1$ and $F$ are fixed and can always be chosen such that $\ker \partial_1 = \im \partial_2$ and $\partial_1$ is surjective.

As in standard surgery, to generically ensure that the distance of the code is preserved, we also want the auxiliary system to be a good expander.
Constructing chain complexes of dimension greater than one with good expansion is a notoriously difficult problem~\cite{lubotzky_high_2018,oppenheim_new_2025,evra_bounded_2016}.
Furthermore, the distance of the auxiliary system (when viewed as a quantum code) needs to be sufficiently high, as this is the property that determines how well syndrome errors can be detected and corrected.
We also want the port maps to be sparse; however, we have some leeway here, as the maps $\partial_3$ and $f_2$ correspond to meta-checks and not to any physically performed gates.
Moreover, the $X$-checks corresponding to $\partial_1$ can be noiselessly inferred from the final readout of the qubits in $E$.

These requirements, along with the commutativity condition for the port functions, make it very difficult to construct suitable auxiliary systems for fast surgery.
For example, Ref.~\cite{baspin_fast_2025} gives a construction for measuring $k$ logical operators with a space overhead of $\tilde{O}\left( n \cdot k \cdot d \right)$.
This essentially corresponds to an exact space-time trade-off compared to standard surgery~\cite{williamson_lowoverhead_2026,ide_faulttolerant_2025}, so while we are able to achieve faster measurements, we have done so at the cost of a significant space overhead.

The procedure for performing fast surgery is similar to that for standard surgery.
We initialise the qubits in $E$ in the $\ket{+}$ state, perform syndrome extraction with the new merged code, measure out the qubits in $E$ in the $X$ basis, and then infer the logical measurement outcome from products of measurement outcomes in $V$.
Unlike in standard surgery, only a single round of syndrome extraction is necessary to achieve fault tolerance.
The formal fault tolerance of this procedure is subtle and depends on the exact analysis used.
In Ref.~\cite{baspin_fast_2025}, it is assumed that the underlying code is single-shot state preparable, which means multiple surgery operations can be chained fault-tolerantly.
Ref.~\cite{cowtan_fast_2025} instead directly analyses the scenario where multiple surgery operations are performed sequentially and shows that these operations can be fault-tolerantly decoded as a block if the compacted distance of these operations is preserved.

\section{Lifted surgery: Theory}\label{sec:ls-theory}
In this section we present lifted surgery formally, establish the possible measurements and addressability, and discuss various subfamilies where the maps have specific forms.

\subsection{Group algebra codes}
Let $G$ be a finite Abelian group and $R = \Ft[G]$ its group algebra over $\Ft$.
We consider codes described by a length-$2$ chain complex
\begin{align*}
    C_\bullet = C_2 \xrightarrow{\partial_2} C_1 \xrightarrow{\partial_1} C_0
\end{align*}
where $C_n \simeq R^{m_n}$ are finite free $R$-modules and the boundary maps $\partial_n$ are $R$-linear.
Expanding each copy of $R$ in the basis indexed by $G$ gives a chain complex of finite-dimensional $\Ft$-vector spaces, defining a binary CSS code.
We call codes obtained in this way \emph{Abelian group algebra codes}.

This framework includes two-block group algebra codes~\cite{lin_quantum_2024}, such as generalised bicycle (GB) codes~\cite{panteleev_degenerate_2021} and bivariate bicycle (BB) codes~\cite{bravyi_highthreshold_2024}, as well as Abelian multi-cycle codes~\cite{lin_abelian_2026} and lifted product codes.
These families provide QLDPC codes with strong finite-length performance~\cite{panteleev_degenerate_2021,bravyi_highthreshold_2024,gu_qgpu_2026}.
\subsection{Block decomposition of the logical space} \label{sec:algebras}
Algebraic decompositions have been used to study the logical space of bivariate bicycle and related group algebra codes~\cite{breuckmann_logical_2025,eberhardt_logical_2024,chen_anyon_2025,postema_existence_2026,lee_logical_2026,sabo_spectral_2026}.
As we consider codes defined over an Abelian group algebra $R= \Ft[G]$, we first recall some standard facts about the structure of $R$, in particular its block decomposition. Importantly, the codespace inherits this block decomposition, which governs the structure of the logical subspace and therefore constrains the possible logical operations.
We refer the reader to \Cref{app:group_algebras} and Refs.~\cite{atiyah_introduction_2018, weibel_introduction_1994} for a more detailed presentation.
\begin{theorem}[Block decomposition of a binary Abelian group algebra]
    Let $G$ be a finite Abelian group and $R = \Ft[G]$.
    Then there are pairwise orthogonal primitive idempotents  $e_1, \ldots, e_\nu \in R$ with
    \begin{align*}
        e_1 + \cdots + e_\nu = 1
    \end{align*}
    such that
    \begin{align*}
        R \simeq \prod_{i=1}^\nu R_i, \quad R_i = e_i R.
    \end{align*}
    Each block $R_i$ is a finite local Artinian ring.
    We write $(R_i, \mm_i)$ for this local ring, where $\mm_i$ is its unique maximal ideal.
    If $|G|$ is odd then each block is a finite field.
\end{theorem}
The idempotents, $e_i$, act as projectors.
For any $R$-module, $M$, we have
\begin{align*}
    M \simeq \bigoplus_{i=1}^\nu e_i M
\end{align*}
and every $R$-linear map $f: M \to N$ preserves this decomposition, since
\begin{align*}
    f(e_i x) = e_i f(x).
\end{align*}
Thus an $R$-linear map cannot mix different blocks.

Applying this to the complex $C_\bullet$, we obtain the subcomplexes
\begin{align*}
    C_{\bullet, (i)} = e_i C_\bullet
\end{align*}
and a corresponding block decomposition of the homology.
\begin{corollary}[Block decomposition of homology]
    For a finite free $R$-complex $C_\bullet$,
    \begin{align*}
        H_n(C_\bullet) \simeq \bigoplus_{i = 1}^\nu H_n(C_{\bullet, (i)}).
    \end{align*}
    Moreover, every $R$-linear chain map acts independently on these homology blocks.
\end{corollary}
In the following, the physical qubits are in one-to-one correspondence with an $\Ft$-basis of $C_1$ and the logical $Z$-operators are represented by elements of the first homology group $H_1(C_\bullet)$:
\begin{align*}
    H_1(C_\bullet) \simeq \ker \partial_1 / \im \partial_2.
\end{align*}
Thus the block decomposition of $H_{1}(C_\bullet)$ induces a corresponding decomposition of the logical $Z$-operator space. We systematically exploit this decomposition throughout the remainder of this work, both to understand the structure of the logical subspace and to efficiently classify our surgery operations.

\subsection{Surgery maps}
A fast surgery operation~\cite{baspin_fast_2025, cowtan_fast_2025} is specified by a chain map $f_\bullet$ from an auxiliary complex $D_\bullet$ to the code complex $C_\bullet$.
The mapping cone of $f_\bullet$ describes the resulting merged code:
\begin{equation*}
    \begin{tikzcd}
    	{D_2} & {D_1} & {D_0} \\
{C_2} & {C_1} & {C_0}    	
\arrow["\partial^D_2", from=1-1, to=1-2]
    	\arrow["f_2", from=1-1, to=2-1]
    	\arrow["\partial^D_1" ,from=1-2, to=1-3]
    	\arrow["f_1", from=1-2, to=2-2]
    	\arrow["f_0", from=1-3, to=2-3]
    	\arrow["\partial^C_2", from=2-1, to=2-2]
    	\arrow["\partial^C_1", from=2-2, to=2-3]
   \end{tikzcd}
\end{equation*}
Explicitly, since $R$ has characteristic two, the mapping cone of $f_\bullet$ is the chain complex given by
\begin{align*}
    \cone(f)_n = D_{n-1} \oplus C_{n}, \\
    \partial_{n}^{\cone} = \begin{pmatrix}
        \partial_{n-1}^D & 0\\
        f_{n-1} & \partial_{n}^C
    \end{pmatrix}.
\end{align*}
The identity $\partial_n^{\cone} \partial_{n+1}^{\cone} = 0$ is equivalent to the chain map condition $f_{n-1}\partial^D_{n}= \partial_n^C f_n$ and so the cone is a well-defined complex.

To identify which logical operators are measured and which remain in the merged code, we use the short exact sequence of chain complexes
\begin{align}\label{eq:cone-ses}
    0 \rightarrow C_\bullet \rightarrow \cone(f)_\bullet \rightarrow D_\bullet[-1]\rightarrow 0 \, ,
\end{align}
where $(D_\bullet[-1])_n = D_{n-1}$~\cite{ide_faulttolerant_2025, weibel_introduction_1994,cowtan_fast_2025}.

Writing $f_*^n$ for the map induced by $f_\bullet$ on $H_n$, we obtain the following from the associated long exact sequence in homology:
\begin{align}
    0 \rightarrow H_1({C_\bullet})/\im f_*^1\rightarrow H_1(\cone(f)) \rightarrow \ker f_*^0 \rightarrow 0.
    \label{eq:homology_cone}
\end{align}
\Cref{eq:homology_cone} illustrates two distinct contributions to the logical  $Z$ operators of the merged code:
\begin{enumerate}[(i)]
    \item surviving logical operators of the base code, described by $H_1(C_\bullet)/\im f_*^1$, since the operators in $\im f_*^1$ become stabilisers and are measured;
    \item additional logical operators introduced by the auxiliary system, described by $\ker f_*^0$.
\end{enumerate}
Thus $\im f_*^1$ is the measured logical subspace that we classify and target in the following sections, while $\ker f_*^0$ accounts for any additional logical qubits when determining the merged-code parameters.
\begin{remark}
    As a sequence of $\Ft$-vector spaces, \cref{eq:homology_cone} splits, giving a non-canonical isomorphism
    \begin{align*}
        H_1(\cone(f)) \simeq H_1(C_\bullet)/\im f_*^1 \oplus \ker f_*^0.
    \end{align*}
    This splitting need not respect the $R$-module structure, which we retain when analysing addressability below.
    \label{remark:splitting}
\end{remark}

We can remove the additional logical operators by including the auxiliary $X$-checks omitted from the diagram above, corresponding to the final term in \cref{eq:fast-surgery}.
At the binary level, we can append a map $\partial_0^D: D_0 \to D_{-1}$ satisfying
\begin{align*}
    \ker \partial_0^D = \im \partial_1^D.
\end{align*}
The resulting auxiliary complex has $H_0(D_\bullet) = 0$, so the $\ker f_*^0$ contribution vanishes, while $H_1(D_\bullet)$ and the measured subspace $\im f_*^1$ are unchanged.
In the merged code, these checks act only on the auxiliary qubits, and we can infer them from their final $X$-basis readout.
Throughout, we assume that the $D_{-1}$ checks are included to eliminate $H_0(D_\bullet)$ and hence the additional logical operators introduced by the auxiliary system.

\subsection{Lifted surgery}
Our goal is to find explicit instances of practical \emph{fast surgery} maps~\cite{baspin_fast_2025,cowtan_fast_2025}.
The $R$-module structure of Abelian group algebra codes naturally suggests restricting our search to $R$-linear chain maps.
This gives us the following general definition:
\begin{definition}  
\label{def:lifted_surgery}
We define a \emph{lifted surgery} map to be any $R$-linear chain map $f: C_\bullet \rightarrow C_\bullet$ between two copies of the same code complex, and we define $\ls_R(C_\bullet)$ to be the set of all lifted surgery maps on the complex $C_\bullet$:
\begin{align*}
    \ls_R(C_\bullet)
    =
    \left\{
    (f_2,f_1,f_0)
    \;\middle|\;
    \begin{aligned}
    &f_n \in \operatorname{End}_R(C_n),\\
    &f_{n-1}\partial_n=\partial_n f_n
    \end{aligned}
    \right\}.\end{align*}
\end{definition}
Each lifted surgery map $f_\bullet$ specifies a surgery operation that measures the set of logical operators $\im f^1_*$.
Since $f_\bullet$ is $R$-linear, the corresponding surgery operation inherits the rich algebraic structure of the Abelian group algebra $R$ (see \Cref{sec:algebras} and \Cref{app:group_algebras}).
We now exploit this structure to decompose and classify lifted surgery maps.  
\begin{theorem} \label{thm:ls_block_decomposition}
    Let $R=\Ft[G]$ and $C_\bullet$ be a chain complex with $C_n \simeq R^{m_n}$. Let $e_1,
    \dots, e_\nu$ be its primitive idempotents so that
    \begin{align*}
        R = \prod_{i=1}^\nu R_i, \quad R_i = e_i R, \quad C_{\bullet, (i)} = e_i C_\bullet
    \end{align*}
    is the block decomposition of $R$ and $C_\bullet$. Then
    \begin{enumerate}[(i)]
        \item $\ls_R(C_\bullet)$ is a finite unital $R$-subalgebra of 
        $$\prod_n \operatorname{End}_R(C_n).$$
        \item The lifted-surgery algebra decomposes blockwise
        \begin{align*}
            \ls_R(C_\bullet) \simeq \prod_i \ls_{R_i}(C_{\bullet, (i)}).
        \end{align*}
    \end{enumerate}
\end{theorem}
\begin{proof}
    \textit{(i)} Consider the $R$-linear map $\eta$:
    \begin{align}
        \eta : \prod_n\operatorname{End}_R(C_n) &\longrightarrow \prod_n \Hom_R(C_n, C_{n-1}) \label{eq:eta} \\
        \prod_n f_n &\longmapsto \prod_n (f_{n-1}\partial_n - \partial_n f_n).\nonumber
    \end{align}
    By the definition of lifted surgery, $\ls_R(C_\bullet) = \ker \eta$ and so it is an $R$-submodule of $\prod_n \operatorname{End}_R(C_n)$.
    It contains the identity chain map and it is closed under composition: $f \circ g = \prod_n (f_n g_n)$. Thus, $\ls_R(C_\bullet)$ is a unital subalgebra. 
    
    \textit{(ii)} By assumption, every map $f_n$ is $R$-linear, and for every central idempotent $e_i$, $f_n(e_i x) = e_i f_n(x)$, and therefore it splits:
    \begin{align*}
        f(x) = f(1 \cdot x) &= f(e_1 x + \cdots + e_\nu x)\\
        &= f(e_1 x)+ \cdots + f(e_\nu x)\\
        &= e_1 f(x) + \cdots + e_\nu f(x),
    \end{align*}
    in short,
    \begin{align*}
        f = \sum_i e_i f
    \end{align*}
    and so any lifted surgery decomposes.
    Conversely, given a tuple of chain maps $f^{(i)}$ on the blocks, we can extend each map by zero on the other blocks:
    \begin{align*}
        \hat f^{(i)} = (0, \dots, f^{(i)}, \dots, 0)
    \end{align*}
    so that they sum to an $R$-linear chain map on $C_\bullet$.
\end{proof}
\Cref{thm:ls_block_decomposition} implies that the set of lifted surgeries is computable (see \Cref{alg:lifted_surgery}).
Expanding the entries of the $R$-linear map $\eta$ over $\Ft$, we can compute its kernel by solving a linear system of equations over $\Ft$.
It is then easy to see that $\ls_R(C_\bullet)$ has cardinality $2^{\dim_{\Ft}\ker \eta}$.
We can apply the same argument, and computation, block-wise.
If $R_i$ is a field, this is ordinary linear algebra over $R_i$. If $R_i$ is a local non-field block, Nakayama's lemma (see \Cref{lemma:nakayama}) gives a minimal generating set for the $R_i$-module.
\begin{algorithm}[h] \label{alg:lifted_surgery}
    \caption{Compute the lifted surgery algebra of $C_\bullet$}
    
    \begin{algorithmic}
        \Require Chain complex $C_\bullet$ of free $R$-modules, $R = \Ft[G]$
        \Ensure Generators of $\mathrm{LS}_R(C_\bullet)$
        
        \State Compute central primitive idempotents $e_i$ of $R$ such that
        \Statex $e_1 + \dots + e_\nu = 1$ 
        \For{each block $R_i = e_i R$}
            \State Expand the map $\eta_i = e_i \eta $ over $\Ft$ \Comment{\cref{eq:eta}}
            \State Compute the kernel, over $\Ft$, of $\eta_i$
            \State $\ls_{R_i}(C_{\bullet, (i)}) \gets \ker_{\Ft} \eta_i$
        \EndFor
        \State $\ls_R(C_\bullet)
        \gets \prod_i \mathrm{LS}_{R_i}(C_{\bullet,i})$
        \State \Return 
        $\ls_R(C_\bullet)$
    \end{algorithmic}
\end{algorithm}
\subsubsection{Addressability}

We now turn to determining which sets of logical measurements lifted surgery maps allow us to address.
\begin{definition}[Measurable maps and addressable modules] \label{def:addressable_modules}
    For a block $R_i$, consider the homology-action map:
    \begin{align*}
        \star_i: \ls_{R_i}(C_{\bullet, (i)}) &\longrightarrow \operatorname{End}_{R_i}(H_{1, (i)})\\
        f_\bullet &\longmapsto f_{*}
    \end{align*}
    The \emph{measurable maps} on the $i$th block are
    \begin{align*}
        \mathcal{M}_i   = \left\{
        g
        \;\middle|\;
        \begin{aligned}
            &g \in \operatorname{End}_{R_i}(H_{1, (i)})\\
            &g \in \im(\star_i)
        \end{aligned}
        \right\}.
    \end{align*}
    We say that an endomorphism $g$ comes from a lifted surgery if it is a measurable map, i.e., if $g \in \mathcal{M}_i$.
    
    We say that an $R_i$-submodule $W$ is an \emph{addressable submodule} if $W = \im g$ for some $g \in \mathcal{M}_i$ and denote the set of addressable submodules by 
    \begin{align*}
        \mathcal{A}_i = \{\im g \mid g \in \mathcal{M}_i\}.
    \end{align*}
\end{definition}
Combining \Cref{thm:ls_block_decomposition} and \Cref{def:addressable_modules}, we have the following.
\begin{restatable}[Addressability of lifted surgery]{theorem}{addressability} \label{thm:addressability}
    The set of logical measurements realisable via lifted surgery is the set of direct sums of blockwise addressable submodules:
    \begin{align*}
        \mathcal{A}(C_\bullet) &= \{\im f_* \mid f \in \ls_R(C_\bullet)\} \\
        &= \left\{\bigoplus_{i = 1}^\nu W_i \mid W_i \in \mathcal{A}_i\right\}.
    \end{align*}
\end{restatable}
\begin{proof}
    By \Cref{thm:ls_block_decomposition}, every lifted surgery $f$ decomposes as $\bigoplus f_{(i)}$ with $f_{(i)}$ a block-local lifted surgery in $\mathcal{\ls}_{R_i}(C_{\bullet, (i)})$ and the decomposition passes to the image on homology (see~\Cref{app:group_algebras}).
    The central idempotent decomposition passes to submodules as well, and since submodules are closed under multiplication by elements of $G$, addressable submodules contain entire orbits.
\end{proof}
Said differently, lifted surgery has two addressability constraints:
\begin{enumerate}[(i)]
    \item (No cross-block addressability). Every measured set splits as a direct sum over the $R_i$ blocks.
    \item (No intra-orbit addressability). Within a block, each measured set is $G$-invariant.
          Hence if a logical operator is measured, its entire $G$-orbit is measured as well.
\end{enumerate}
Within this addressability regime, we can fully characterise the addressable submodule of a given code, and therefore study the logical operations that lifted surgery enables. 
\begin{restatable}[Lifted surgery and addressable module]{theorem}{addressablemodule}
    Let $W \subseteq H_{1, (i)}$ be an $R_i$-submodule. The following are equivalent:
    \begin{enumerate}
        \item $W \in \mathcal{A}_i$ is a module addressable by a lifted surgery on the $i$th block.
        \item There exists $\phi \in \operatorname{End}_{R_i}(Z_i)$ such that
            \begin{align*}
              \phi(B_i) \subseteq B_i, \quad W = \pi(\phi(Z_i)).
            \end{align*}
        \item There exists $g \in \operatorname{End}_{R_i}(H_{1, (i)})$ with $\im g = W$ and there exists an $R_i$-linear lift $\phi: Z_i \rightarrow Z_i$ such that 
            \begin{align*}
                \pi \phi = g \pi.
            \end{align*}
        \item There exist a submodule $N \subseteq H_{1, (i)}$ and an $R_i$-linear map $\phi: Z_i \rightarrow Z_i$ such that
            \begin{align*}
                \phi^{-1}(B_i) = \pi^{-1}(N), \quad W =\pi(\phi(Z_i)).
            \end{align*}
    \end{enumerate}
    Equivalently, the addressable submodules are
    \begin{align*}
        \mathcal{A}_i   = \left\{
        \pi(\phi(Z_i))
        \;\middle|\;
        \begin{aligned}
            &\phi \in \operatorname{End}_{R_i}(Z_i)\\
            &\phi(B_i) \subseteq B_i
        \end{aligned}
        \right\}.
    \end{align*}
\label{thm:ls_equivalent}
\end{restatable}
\begin{proof}
    See \hyperref[proof:addressable_module]{proof of \Cref{thm:ls_equivalent}} in \Cref{app:main_proofs}.
\end{proof}
\Cref{thm:ls_equivalent} characterises the scope of lifted surgery and motivates the following definition.
\begin{definition}[Granularity of the homology blocks] \label{def:fine}
    We call a homology block $H_{1, (i)}$ \emph{completely fine} if every $\Ft$-subspace of $H_{1, (i)}$ is addressable; \emph{fine} if every $R_i$-submodule of $H_{1, (i)}$ is addressable; and \emph{coarse} if it is not fine.

    In particular, on a completely fine block, each (binary) logical subspace is addressable via lifted surgery.
\end{definition}
Completely fine blocks are fine.
The converse is false in general: measured sets are always $R_i$-submodules, so a fine block is completely fine if and only if every $\Ft$-linear subspace of $H_{1,(i)}$ is an $R_i$-submodule, e.g., when $R_i \simeq \Ft$.

The corollary below gives a sufficient condition for fineness: whenever the blocks of $R$ are fields, every $R_i$-linear endomorphism is measurable and every $R_i$-submodule is addressable.
\begin{restatable}[Field blocks are fine]{corollary}{surgeryfield} \label{cor:full_addressability}
    If $R_i = \kk_i$ is a field, then 
    \begin{align*}
       \mathcal{M}_i = \operatorname{End}_{\kk_i}(H_{1, (i)})
    \end{align*}
   and every $\kk_i$-linear subspace of $H_{1,(i)}$ is addressable. 
\end{restatable}
\begin{proof}
    See \hyperref[proof:surgery_field]{proof of \Cref{cor:full_addressability}} in \Cref{app:main_proofs}.
\end{proof}
We remark that, if $R_i$ is not a field, \Cref{thm:addressability} still holds, and 
\begin{align*}
    \mathcal{A}_i = \{\pi(\phi(Z_i)) \mid \phi(B_i) \subseteq B_i\}.
\end{align*}
However, not all $R_i$-linear endomorphisms of $H_{1, (i)}$ lift to a map $\phi: Z_i \rightarrow Z_i$ preserving the boundaries.
Unlike the field case, the quotient condition $W \simeq H_{1, (i)}/N$ is not, by itself, sufficient for liftability and we must also check the compatibility condition $\phi(B_i) \subseteq B_i$.
In~\Cref{app:ext} we recast the compatibility condition in terms of the extension class of the cycle sequence in $\mathrm{Ext}_{R_i}^1$ and derive a more compact characterisation of the addressable submodules over non-field blocks.

A natural way to enlarge the class of addressable modules is to allow for an auxiliary complex $D_\bullet \neq C_\bullet$.
This would enlarge the space of admissible chain maps and could, in principle, make the compatibility condition easier to satisfy while keeping the set of surgeries computable. However, as long as one restricts to $R$-linear chain maps, the same addressability constraints remain. 
Of course, allowing $\Ft$-linear chain maps would remove this limitation~\cite{benhemou_automated_2026}.

\subsection{Structured families of lifted surgeries}
In this section we describe three sub-families of lifted surgery, ordered by the amount of additional structure they require from the chain complex $C_\bullet$.
These families are of particular interest because their structure provides a natural framework for constructing merged codes that remain LDPC. 
\subsubsection{Scalar surgery}\label{sec:scalar-surgery}
We start with scalar surgery, which is the limiting case in which a single element of the group algebra describes the chain maps.
\begin{definition}[Scalar surgery]
  A \emph{scalar surgery} map is a lifted surgery map of the form $f_\bullet = \mu_c$, where $\mu_c$ is multiplication by an element $c \in R$.
\end{definition}
\begin{lemma}[Scalar multiplication is a chain map] \label{lemma:scalar_multiplication}
  Let $R$ be a commutative ring and $C_\bullet$ a chain complex of
  $R$-modules. For every $c \in R$, multiplication by $c$ at every degree,
  \begin{align*}
      \mu_c: C_n &\longrightarrow C_{n}\\
      x &\longmapsto c x
  \end{align*}
  defines a chain map $\mu_c : C_\bullet \rightarrow C_\bullet$.
  The induced map on homology is again multiplication by $c$, and it preserves the block decomposition, acting on $H_{n,(i)}$ as multiplication by $c_i = c e_i$.
\end{lemma}
\begin{proof}
  The differentials are $R$-linear and $R$ is commutative; hence they commute with multiplication by $c$.
  So $\mu_c$ satisfies the chain map condition and $[z] \longmapsto [cz]$ is well defined on homology.
  Last, $\mu_c$ commutes with each central idempotent $e_i$, so it restricts to the block subcomplex $e_iC_\bullet$, acting there as $\mu_{c_i}$.
\end{proof}

The motivation for introducing scalar surgery is that, for certain codes such as two-block group algebra codes, it preserves desirable structure in the code.
More specifically, for codes that can be described by a Koszul complex, scalar surgery is closed.
\Cref{def:koszul} shows that the mapping cone of $\mu_c$ on $K_\bullet(a,b)$ \emph{is} the length-$3$ Koszul complex $K_\bullet(a,b,c)$.
Thus, scalar surgery preserves the Koszul form and extends the defining sequence by the surgery element $c$.
We develop this description in \Cref{sec:scalar-surgery-koszul}.
This is desirable since codes that can be described by Koszul complexes often have very efficient hardware implementations, specifically on systems with reconfigurable qubits~\cite{viszlai_matching_2025,menon_magic_2026}.

We study the action of scalar surgery block-wise.
For $c \in R$, let $c_i = e_i c \in R_i$.
On the $i$th block $(R_i, \mm_i)$ (see \Cref{thm:block_algebra}), the scalar surgery $\mu_c$ measures the following:
\begin{enumerate}[(i)]
    \item The whole block $H_{1, (i)}$ if $c_i$ is invertible in $R_i$;
    \item The submodule $c_i H_{1,(i)} \subseteq \mm_i H_{1, (i)}$ if $c_i \neq 0$ is nilpotent; the measured set is non-zero if and only if $c_i \not \in (0:_{R_i} H_{1, (i)})$, 
    where
    \begin{align*}
        (0:_R M) = \{x \in R \mid xm = 0 \text{ for all } m \in M\}
    \end{align*}
    is the \emph{annihilator} of the module $M$ in $R$.
    In particular, letting $\rho(c_i) = \max \{t \mid c_i \in \mm_i^t\}$, the scalar surgery $\mu_c$ measures the submodule
    \begin{align*}
        c_i H_{1, (i)} \subseteq \mm_i^{\rho(c_i)} H_{1, (i)}.
    \end{align*}
    See \Cref{app:group_algebras}.
    \item Nothing if $c_i$ is zero in $R_i$.
\end{enumerate}
The three cases above cover all possibilities for scalar surgery because in a local Artinian ring, $(R_i, \mm_i)$, every element is either a unit or nilpotent \cite{atiyah_introduction_2018}.

Lastly, we show that if the homology group $H_{1, (i)}$ is cyclic (i.e., the group is generated by one element), then scalar surgery covers all possible measurable maps.
\begin{restatable}[Structure of scalar surgery for cyclic homology]{proposition}{surgerycyclic} \label{prop:cyclic_scalar}
    Fix a block $(R_i, \mm_i)$ and assume that $H_{1, (i)}$ is a cyclic $R_i$-module.
    Then the following hold:
    \begin{enumerate}[(i)]
        \item All measurable maps $\mathcal{M}_i$ are scalar surgeries.
        \item If we further assume that $R_i$ is a chain ring (i.e., all its ideals are totally ordered by inclusion), every $R_i$-submodule of $H_{1, (i)}$ is addressable and can be measured via scalar surgery.
        Moreover, annihilator shift and unit rescaling yield all possible measurement-equivalent surgeries.
    \end{enumerate}
\end{restatable}
\begin{proof}
    See \hyperref[proof:cyclic_scalar]{proof of \Cref{prop:cyclic_scalar}} in \Cref{app:main_proofs}.
\end{proof}
In other words, a cyclic homology block over a chain ring is fine (cf.~\Cref{def:fine}): scalar surgeries alone realise every addressable submodule.
By contrast, when $H_{1,(i)}$ is not cyclic, scalar surgeries need not exhaust the lifted surgeries: there can be measurable maps that are not scalar.

\subsubsection{Diagonal surgery}
We now introduce diagonal surgery as the most natural generalisation of scalar surgery.
While scalar surgery considers chain maps given by the same scalar multiple of the identity at every level, diagonal surgery allows for chain maps described by arbitrary diagonal maps.
\begin{definition}[Diagonal surgery]
    A lifted surgery $f_\bullet \in \ls_R(C_\bullet)$ is \emph{diagonal} if every $f_n$ is a diagonal matrix over $R$:
    \begin{align*}
        f_n = \mathrm{diag}(c_{n,1}, \dots, c_{n,m}), \quad c_{n, j} \in R.
    \end{align*}
\end{definition}
\begin{proposition}[Existence of diagonal surgery] \label{prop:diagonal_surgery}
    Given a length-$n$ complex $C_\bullet$, the diagonal maps
    \begin{align*}
        f_i = \mathrm{diag}(c_{i,1}, \dots, c_{i,m})
    \end{align*}
    form a chain map $f_\bullet = (f_0, \dots, f_n)$ if and only if, for every $n$ and any pair of indices $(i, j)$,
    \begin{align*}
        c_{n-1, i} + c_{n, j} \in (0:_R d_{i,j}^{(n)})
    \end{align*}
    where $d_{i, j}^{(n)}$ is the $(i, j)$-entry of $\partial_n$.
\end{proposition}
\begin{proof}
    By definition of a chain map $f_\bullet$, $f_{n-1}\partial_n = \partial_n f_n$.
    Entrywise, for a diagonal map, this holds if and only if
    \begin{align} \label{eq:chain_diago}
        c_{n-1, i} d_{i, j}^{(n)} = c_{n, j} d_{i, j}^{(n)}.
    \end{align}
  Over a commutative $\Ft$-algebra $R$, \cref{eq:chain_diago} holds if and only if $c_{n-1, i} + c_{n, j} \in (0:_R d_{i,j}^{(n)})$, as desired.
\end{proof}
\Cref{prop:diagonal_surgery} says that diagonal surgery is of interest when the entries of the differentials are zero divisors.
When this is not the case, and $(0 :_{R} d_{i, j}^n) =0$, we recover scalar surgery.
In fact, equal coefficients satisfy
\begin{align*}
    c_{n-1, i} + c_{n, j} = 0 \quad \text{in } R
\end{align*}
as $R$ has characteristic $2$.

\subsubsection{Lifted product surgery}\label{sec:lifted-prod-surgery}
We now specialise to tensor products of $R$-complexes.
Take the free modules
\begin{align*}
   A_0 = R^{n_A}, \quad A_1 = R^{m_A} 
\end{align*}
and the $R$-linear map $\partial_A: A_1 \rightarrow A_0$, represented by an $n_A \times m_A$ matrix over $R$ that we again denote by $\partial_A$.
Similarly for $B_0, B_1$ and $\partial_B$. 
We describe a lifted product code using the following chain complex:
\begin{equation*} 
    \begin{tikzcd}[column sep=13, row sep=10]
        & {A_0 \otimes B_1}    & &  \\
        {A_1 \otimes B_1}      & & {A_0\otimes B_0} \\
        & {A_1 \otimes B_0}
        \arrow[from=2-1, to=1-2, "\partial_A \otimes \idty_{m_B}"]
        \arrow[from=1-2, to=2-3, "\idty_{n_A} \otimes \partial_B"]
        \arrow[from=2-1,to=3-2,"\idty_{m_A} \otimes \partial_B "' ]
        \arrow[from=3-2, to=2-3, "\partial_A \otimes \idty_{n_B}"']
    \end{tikzcd}
\end{equation*}
where
\begin{align*}
    H_Z &= \left( \begin{array}{c|c} \partial_A^* \otimes \idty_{m_B}& \idty_{m_A} \otimes \partial_B^*\end{array} \right) \\
    H_X &= \left( \begin{array}{c|c} \idty_{n_A} \otimes \partial_B & \partial_A \otimes \idty_{n_B}\end{array} \right) \, .
\end{align*}
Here we take the tensor product $\otimes$ over $R$, i.e., $\otimes = \otimes_R$, and for the element $r = \sum_{g \in G} r_g g \in R$, $r_g \in \Ft$, we define its conjugate $r^*=\sum_{g\in G} r_g g^{-1}$.
For a matrix $M$ with coefficients in $R$, $(M^*)_{ij} = (M_{ji})^*$.

We take an auxiliary copy of the base code and we consider the following diagram:
\begin{equation} \label{eq:lifted_product_surgery}
    \begin{tikzcd}[column sep=13, row sep=10, ampersand replacement=\&]
        \& {A_0 \otimes B_1} \& \\
        {A_1 \otimes B_1} \&  \& {A_0 \otimes B_0} \\
        \& {A_1 \otimes B_0} \& \\[8em]
        \& {A_0 \otimes B_1} \& \\
        {A_1 \otimes B_1} \&  \& {A_0 \otimes B_0} \\
        \& {A_1 \otimes B_0}
        \arrow[from=2-1, to=1-2]
        \arrow[from=2-1, to=3-2]
        \arrow[from=1-2, to=2-3]
        \arrow[from=3-2, to=2-3]
        \arrow[from=5-1, to=4-2]
        \arrow[from=5-1, to=6-2]
        \arrow[from=4-2, to=5-3]
        \arrow[from=6-2, to=5-3]
        \arrow[
            dashed,
            from=2-1,
            to=5-1,
            "\alpha_1 \otimes \beta_1" description,
            shorten <=20pt,
            shorten >=20pt
        ]
        \arrow[
            dashed,
            from=3-2,
            to=4-2,
            "{\begin{pmatrix}
                \alpha_0 \otimes \beta_1 & 0 \\
                0 & \alpha_1 \otimes \beta_0
            \end{pmatrix}}" description,
            shorten <=4pt,
            shorten >=4pt
        ]
        \arrow[
            dashed,
            from=2-3,
            to=5-3,
            "\alpha_0 \otimes \beta_0" description,
            shorten <=20pt,
            shorten >=20pt
        ]
    \end{tikzcd}
\end{equation}
The diagram is valid when both squares commute, or equivalently, whenever the vertical arrows define a chain map between the two copies of the lifted product complex.
Crucially, in the diagram in \cref{eq:lifted_product_surgery}, the middle component preserves the decomposition of the middle space $\big(A_0 \otimes B_1\big) \oplus \big(A_1 \otimes B_0\big)$.

Expanding the first square gives
\begin{align*}
\partial_A\alpha_1 \otimes \beta_1 &= \alpha_0\partial_A \otimes \beta_1,\\
\alpha_1 \otimes\partial_B \beta_1 &=\alpha_1 \otimes \beta_0 \partial_B,
\end{align*}
and expanding the second square we have
\begin{align*}
\alpha_0 \otimes \partial_B\beta_1 &= \alpha_0 \otimes \beta_0 \partial_B,\\
\partial_A \alpha_1 \otimes \beta_0 &=\alpha_0\partial_A \otimes \beta_0.
\end{align*}
Thus a simple set of sufficient conditions is
\begin{align*}
    \partial_A \alpha_1 &= \alpha_0 \partial_A, \nonumber\\
    \partial_B \beta_1 &= \beta_0 \partial_B.
\end{align*}
In other words, it is sufficient for $\alpha=(\alpha_1,\alpha_0)$ and $\beta=(\beta_1,\beta_0)$ themselves to be chain maps on the two constituent complexes:
\begin{equation*}
    \begin{tikzcd}
        A_1 & A_0\\
        A_1 & A_0
        \arrow[from=1-1, to=1-2, "\partial_A"]
        \arrow[from=2-1, to=2-2]
        \arrow[dashed, from=1-1, to=2-1, "\alpha_1" description]
        \arrow[dashed, from=1-2, to=2-2, "\alpha_0" description]
    \end{tikzcd},
    \quad\quad
    \begin{tikzcd}
        B_1 & B_0\\
        B_1 & B_0
        \arrow[from=1-1, to=1-2, "\partial_B"]
        \arrow[from=2-1, to=2-2]
        \arrow[dashed, from=1-1, to=2-1, "\beta_1" description]
        \arrow[dashed, from=1-2, to=2-2, "\beta_0" description]
    \end{tikzcd}
\end{equation*}
where the maps $\alpha_i$ and $\beta_i$ are $R$-linear maps. 

Formally, we have the following two definitions.
\begin{definition}[Chain pair] \label{def:chain_pairs}
    Let $A_\bullet : A_1 \xrightarrow[]{\partial_A} A_0$ be an $R$-complex.
    A \emph{chain pair} on $A_\bullet$ is a pair of $R$-linear maps $\alpha = (\alpha_1, \alpha_0), \alpha_n \in \operatorname{End}_R(A_n)$ such that $\partial_A \alpha_1 = \alpha_0 \partial_A$.
    We write $\mathcal{E}_A$ for the set of all chain pairs of the complex $A_\bullet$.
\end{definition}
\begin{definition}[Lifted product surgery] \label{def:lifted_prodcut_surgery}
    Given two chain pairs $\alpha$ and $\beta$, we define a \emph{lifted product surgery map} as the lifted surgery map defined as in \cref{eq:lifted_product_surgery}.
    For a lifted product surgery derived from chain pairs $\alpha$ and $\beta$, we write
    \begin{align*}
        \alpha \otimes \beta = \begin{pmatrix}
            \alpha_1 \otimes \beta_1, &\begin{pmatrix}
             \alpha_0 \otimes \beta_1 & 0 \\
             0 & \alpha_1 \otimes \beta_0
            \end{pmatrix},
            & \alpha_0 \otimes \beta_0
        \end{pmatrix}.
    \end{align*}
\end{definition}
Combining \Cref{def:chain_pairs} and \Cref{thm:ls_block_decomposition}, it is immediate to see that $\mathcal{E}_A$ is a unital $R$-subalgebra of $\operatorname{End}_R(R^{n_A}) \times \operatorname{End}_R(R^{m_A})$, and in particular block-wise computable.
Moreover, lifted product surgery is closed under multiplication:
\begin{align*}
   ( \alpha \otimes \beta ) \circ (\alpha' \otimes \beta') = (\alpha \circ \alpha') \otimes (\beta \circ \beta').
\end{align*}
As an example, let $\lambda \in R$ and $v_i, w_i \in A_i$ such that
\begin{align*}
    \partial_A v_1 = \lambda v_0,\\
    w_0^T \partial_A = \lambda w_1^T.
\end{align*}
Then $(v_1w_1^T,v_0 w_0^T)$ is a chain pair for $\partial_A$.
In fact,
\begin{align*}
    \partial_A(v_1w_1^T) = (\lambda v_0) w_1^T = v_0 (w_0^T \partial_A) = (v_0 w_0^T)\partial_A.
\end{align*}
Combining this chain pair with the identity on $B_\bullet$ gives a lifted product surgery.
\begin{remark}
    Hypergraph product (HGP) codes~\cite{tillich_quantum_2014} are lifted product codes over the group algebra $R=\Ft$.
    The block decomposition of $R = \Ft$ consists of a single field block so, by \Cref{cor:full_addressability}, the homology is \emph{completely fine} and every binary subspace of the logical space is addressable.
    The relevant question for HGP codes is therefore not whether a logical measurement can be realised but whether the merged code has good properties.
    
    Lifted product surgery on HGP codes can measure any tensor-product space of the form
    \begin{align*}
        W = \bigl( \im \alpha_0 \otimes \im \beta_1\bigr) \oplus \bigl( \im \alpha_1 \otimes \im \beta_0 \bigr)
    \end{align*}
    for a total of 
    \begin{align*}
    \rank(\alpha_0) \rank(\beta_1) + \rank(\alpha_1) \rank(\beta_0)    
    \end{align*}
     measured logical operators. 
     Here, with a slight misuse of notation, we write, e.g., $\alpha_0$ for the induced map on homology $(\alpha_0)_*$.
     In particular, single logical qubits can be addressed by rank-one chain pairs.
     Generic subspaces, such as a logical that spans the two sectors, can instead be measured via general lifted surgery (\Cref{def:lifted_surgery} and \Cref{cor:full_addressability}).
     
     We can interpret the method proposed in Ref.~\cite{chang_constanttime_2026} as a form of lifted product surgery in which the auxiliary code is not a copy of the base code, but rather a subcode.
     In this formulation, one of the two chain-map pairs is $(\idty,\idty)$, while the other has the form
     \begin{align*}
         (\rank 1, 0) \quad\text{ or }\quad (0, \rank 1).
     \end{align*}
    Their scheme measures equivalent logical operators across rows or columns in parallel, with rigorous LDPC, distance preservation, and fault-tolerance guarantees in the sense of Ref.~\cite{cowtan_fast_2025}.
    Our lifted surgery framework offers a way to individually target logical qubits, and hence offers a potential solution to an open problem in that work, albeit without the rigorous guarantees on sparsity and distance that their work provides.
\end{remark}

\subsection{Surgery equivalence} \label{sec:surgery-equivalence}
Different lifted surgery maps in $\ls_R(C_\bullet)$ may implement the same logical measurement on $C_\bullet$ but exhibit very different properties when we look at their associated cones over $\Ft$.
Since the sparsity and distance of the merged code are critical to the fault-tolerance of the surgery procedure, understanding how to find equivalent lifted surgeries that yield good code properties when we lift the merged chain complex to $\Ft$ is essential.
\begin{definition}[Homotopic equivalence]
    Two lifted surgeries $f_\bullet, f'_\bullet$ are \emph{homotopic}, $f_\bullet\simeq f'_\bullet$, if there exist $R$-linear maps $h_n :C_{n} \longrightarrow C_{n+1}$ such that
    \begin{align*}
        f_n - f_n' = \partial_{n+1}h_{n}-h_{n-1}\partial_n.
    \end{align*}
    In short, $f_\bullet\simeq f'_\bullet$ if $ f - f' = \partial h - h \partial$.
\end{definition}
\begin{restatable}{proposition}{homotopicequivalence} \label{prop:cone_homotopy}
    Let $f, f' : D_\bullet \rightarrow C_\bullet$ be chain maps. If $f$ and $f'$ are homotopic then $f_* = f_*'$ on every homology group and the chain complexes $\cone(f)$ and $\cone(f')$ are isomorphic.
\end{restatable}
\begin{proof}
    See~\Cref{app:main_proofs}.
\end{proof}
The merged codes built from homotopic lifted surgery maps are isomorphic over $R$, and therefore have the same physical and logical qubit counts $[[n_\textrm{merged}, k_\textrm{merged}]]$.
Their binary boundary maps, however, may have different row and column weights.
Indeed, the weights of binary codes are invariant only under qubit permutations, which over $R$ correspond to monomial matrices with entries in $G$.
The isomorphism of the cone code induced by a homotopy is not, in general, of this form and lifts over $\Ft$ to a general invertible matrix.
Invertible matrices correspond to CNOT circuits which may change the weights of the code.

Homotopy is therefore a powerful feature of lifted surgery.
Given a representative $f$ of a lifted surgery, we can search within its homotopy class for $f' \simeq f$ such that the merged code has improved parameters in terms of code distance and qubit/check degree.
See \Cref{alg:optimise}.
\begin{definition}[Measurement equivalence of scalar surgery] \label{def:meas_equiv}
    Two scalar surgeries $\mu_c$ and $\mu_{c'}$ are \emph{measurement-equivalent} if they define the same addressable submodule:
     \begin{align*}
        c H_1(C_\bullet) = c' H_1(C_\bullet).
    \end{align*}
\end{definition}
The lemma below exhibits two simple sufficient conditions for measurement equivalence of scalar surgery.
\begin{lemma}[Equivalence of scalar surgery] \label{lemma:meas_equivalence}
    Given two scalar surgeries $\mu_c$ and $\mu_{c'}$, each of the following implies measurement equivalence:
    \begin{enumerate}[(i)]
        \item Annihilator shift: $c + c' \in (0:_R H_1(C_\bullet))$; in this case $(\mu_c)_*^1 = (\mu_{c'})_*^1$ as maps on $H_1(C_\bullet)$.
        \item Unit rescaling: $c' = uc$ with $u \in R$ invertible; in this case the induced maps on $H_1$ differ by the automorphism $(\mu_u)_*$.
    \end{enumerate}
\end{lemma}
\begin{proof}
    Immediate from \Cref{def:meas_equiv}.
\end{proof}
Measurement equivalence is, in general, weaker than homotopy, which forces equality of the induced maps at all homology levels $H_{n}(C_\bullet)$.

In \Cref{alg:optimise} below, the term \emph{beats} refers to a chosen figure of merit of the binary merged code, usually check and qubit degree or distance. 
\begin{algorithm}[t] \label{alg:optimise}
    \caption{Lifted surgery optimisation}
    \begin{algorithmic}[1]
        \Require Sample budget $N$
        \Statex A lifted surgery $f_\bullet \in \mathrm{LS}_R(C_\bullet)$ \Comment{ \Cref{alg:lifted_surgery}}
        \Ensure A surgery $f'_\bullet$ equivalent to $f_\bullet$ 
        \Statex with weight-optimised merged code $\cone(f'_\bullet)$
        \State $f'_\bullet \gets f_\bullet$
        \If{$f_\bullet = \mu_c$ } \Comment{scalar surgery}
            \For{$s = 1, \dots, N$}
                \For{each block $R_i$}
                    \State sample $u_i \in R_i^\times$ and $a_i \in (0:_{R_i} H_{1,(i)})$
                    \State $c_i' \gets u_i c_i + a_i$
                \EndFor
                \State $c' \gets \sum_i c_i'$; \,
                \If{$\cone(\mu_{c'})$ beats $\cone(f'_\bullet)$}
                \State set $f'_\bullet \gets \mu_{c'}$
                \EndIf
            \EndFor
        \Else
         \Comment{generic surgery}
            \For{$s = 1, \dots, N$}
                \State sample $h_0 \in \Hom_R(C_0, C_1)$,
                \Statex sample $h_1 \in \Hom_R(C_1, C_2)$
                \State $f_0'' \gets f_0' + \partial_1 h_0$
                \State $f_1'' \gets f_1' + \partial_2 h_1 + h_0 \partial_1$
                \State $f_2'' \gets f_2' + h_1 \partial_2$
                \If {$\cone(f''_\bullet)$ beats $\cone(f'_\bullet)$}
                \State set $f'_\bullet \gets f''_\bullet$
                \EndIf
            \EndFor
        \EndIf \\
        \Return $f'_\bullet$
    \end{algorithmic}
\end{algorithm}

\section{Lifted Surgery: Constructions}\label{sec:ls-constructions} \label{sec:constructions}
In the previous sections, we introduced lifted surgery formally, showed that the set of lifted surgery maps is a computable space of chain maps and described concrete methods for finding good representatives via homotopy and measurement equivalence.
\Cref{alg:lifted_surgery} and~\Cref{alg:optimise} summarise the resulting methods.
In this section, we take a constructive perspective and show how to implement lifted surgery explicitly in concrete code instances.

\Cref{tab:surgery_summary} summarises the results.
We determine the overheads by noting that for a code consisting of $n$ data qubits, $m_X$ $X$-type checks, and $m_Z$ $Z$-type checks, we obtain a merged code with $n + m_X$ data qubits, $m_Z + n$ $Z$-type checks, and $m_X$ $X$-type checks (discounting any $X$-type checks inferred during readout), when performing a $Z$-type logical measurement.

\begin{remark}\label{remark:lifted-surgery-families}
The notion of a code \emph{family} commonly used in the literature -- e.g., the family of bivariate bicycle codes -- is in general too restrictive for the subsequent analyses.
For a fixed \emph{binary} code $C$ represented over a group algebra, $R = \Ft[G]$, the measurable chain maps (see \Cref{def:addressable_modules}) determine the set of lifted surgery operations and are $R$-linear by definition.
If the same binary code $C$ admits an alternative representation over a different group algebra, $R' = \Ft[G']$, the corresponding set of $R'$-linear maps can be different.
As discussed below the proof of \Cref{thm:ls_block_decomposition}, the size of the orbits of the action of the group $G$ controls the addressability of a lifted surgery map.
Different groups can therefore lead to different orbit structures. 
Hence, the same binary code may support different sets of lifted surgery operations depending on its group algebra representation.
\end{remark}

\subsection{Algebraic tools for block decomposition} \label{sec:2BGACs-block-decomposition}
In \Cref{sec:algebras} we described the block decomposition of $R$ abstractly and showed that it induces a corresponding block decomposition of the logical space.
The Koszul framework extends beyond two-block codes to constructions such as trivariate codes~\cite{jacob_singleshot_2026}, and allows local blocks that are not fields.
Within this framework, \Cref{thm:block_koszul_homology} identifies the blocks supporting logical operators directly from local algebraic data.
We now describe how to concretely compute the blocks of a Koszul complex from a polynomial representation.
For
\begin{align*}
    G = \mathbb{Z}_{\ell_1} \times \ldots \times \mathbb{Z}_{\ell_t}
\end{align*}
we write its group algebra as $R = B/I$, where
\begin{align*}
    B = \Ft[x_1, \dots, x_t] \quad \text{and}\quad I = (x_1^{\ell_1}-1, \dots, x_t^{\ell_t}-1).
\end{align*}
We write the primary decomposition of the ideal $I$ as
\begin{align*}
    I = T_1 \cap \cdots \cap T_\nu
\end{align*}
where $T_i$ are its primary components.
This gives a concrete realisation of the blocks from \Cref{sec:algebras} as
\begin{align*}
    R_i \simeq B / T_i.
\end{align*}
For each primary component we define the radical as
\begin{align*}
    P_i = \sqrt{T_i}.
\end{align*}
Because $T_i$ is primary, $P_i$ is a prime ideal.
Since $I$ is a zero-dimensional ideal, $R_i$ is local and $\mm_i = P_i/T_i$ is its unique maximal ideal.
This is useful because a polynomial $f \in B$ becomes a non-unit on the $i$th block when $f \in P_i$.

For a Koszul complex, $K_\bullet(a, b)$, define the ideal
\begin{align*}
    J = (a, b).
\end{align*}
The homology is supported on the $i$th block exactly when both $a$ and $b$ are non-units in $R_i$, or equivalently, when
\begin{align*}
    J \subseteq P_i.
\end{align*}
We define the surviving part of the block as
\begin{align*}
    Q_i = B/(T_i + J).
\end{align*}
Then
\begin{align*}
    J \subseteq P_i \iff Q_i \neq 0,
\end{align*}
as in \Cref{thm:block_koszul_homology}.
Each Koszul code example in the subsequent section follows the same procedure: compute the primary components $T_i$, use $P_i$ to identify the supported blocks, inspect the corresponding quotients $Q_i$, and then find sparse representatives for each class of lifted surgery map.
We compute the decompositions using specialised algorithms based on Gr\"obner bases~\cite{gianni_grobner_1988, decker1999primary, cox_ideals_2007}.
\begin{table*}[tb]
    \centering
    \renewcommand{\arraystretch}{1.25}
    \setlength{\tabcolsep}{2pt}
    \begin{tabular*}{\textwidth}{@{\extracolsep{\fill}}lllllll@{}}
        \toprule
        base code & type & description & rank & merged code & $\omega_c$ & $\omega_q$ \\
        \midrule
        GB $\llbracket 210, 10, 10\rrbracket$
        & scalar
        & $1 + x + x^{29}$
        & $6$
        & $\llbracket 315, 4, \leq 10\rrbracket$
        & $9$
        & $9$ \\
        & scalar
        & $1 + x^3 + x^{29}$
        & $4$
        & $\llbracket 315, 6, \leq 10\rrbracket$
        & $9$
        & $9$ \\
        \midrule
        BB $\llbracket 144, 12, 12\rrbracket$
        & scalar
        & $1 + xy + x^3 y^2 + x^5 y^4$
        & $4$
        & $\llbracket 216, 8, \leq 10\rrbracket$
        & $10$
        & $11$ \\
        & scalar
        & $1 + xy + x^2 y^3 + x^6 y$
        & $8$
        & $\llbracket 216, 4, \leq 11\rrbracket$
        & $10$
        & $11$ \\
        & lifted
        & any rank-$1$ subspace of $W$ in~\cref{eq:gross-subspace}
        & $1$
        & $\llbracket 216, 11, \leq 12\rrbracket$
        & $14$
        & $14$ \\
        \midrule
        GHP $\llbracket 882, 48, 16\rrbracket$
        & scalar
        & $(1 + x) + (x^9 + x^{38})y + (1 + x)y^3$
        & $12$
        & $\llbracket 1323, 36, \leq 16\rrbracket$
        & $14$
        & $15$ \\
        & scalar
        & $(1 + x^{10}) + (x^7 + x^{43})y + (1 + x^{10})y^5$
        & $36$
        & $\llbracket 1323, 12, \leq 16\rrbracket$
        & $14$
        & $15$ \\
        \midrule
        QRC $\llbracket 90, 8, 10\rrbracket$
        & lifted
        & any rank-$1$ subspace ($15$ orbits)
        & $1$
        & $\llbracket 135, 7, \leq 10\rrbracket$
        & $11$
        & $11$ \\
        & lifted
        & any rank-$2$ subspace ($380$ orbits)
        & $2$
        & $\llbracket 135, 6, \leq 10\rrbracket$
        & $12$
        & $12$ \\
        \midrule
        QRC $\llbracket 198, 8, 16\rrbracket$
        & lifted
        & any rank-$1$ subspace ($45$ orbits)
        & $1$
        & $\llbracket 297, 7, \leq 16\rrbracket$
        & $14$
        & $14$ \\
        \bottomrule
    \end{tabular*}
    \caption{Summary of optimised surgery constructions.
    For scalar constructions, the description gives the corresponding algebra element.
    For more general lifted surgery constructions, which we obtain by numerical search, it describes the accessible logical subspaces.
    Here rank denotes the number of independent logical operators measured, and $\omega_c$ and $\omega_q$ denote the maximum check weight and qubit degree of the merged code.    
    }
    \label{tab:surgery_summary}
\end{table*}

\subsection{$\llbracket 90,8,10 \rrbracket$ code}\label{sec:ls-constructions-qrc}
Consider the $\llbracket 90,8,10 \rrbracket$ quantum radial code defined as a lifted product code over $R = \Ft[x]/(x^5 - 1)$.
Its protographs are
\begin{align*}
    a = \begin{pmatrix}
        x^2 & x^3 & x^4 \\
        x & x^4 & x \\
        x^4 & x^3 & x^2
    \end{pmatrix},
    \quad
    b = \begin{pmatrix}
        x^2 & x^2 & 1 \\
        x^4 & 1 & x^4 \\
        x & x^3 & 1
    \end{pmatrix}.
\end{align*}
As discussed in~\Cref{remark:lifted-surgery-families}, the available surgeries depend on the group presentation.
The same code has a bivariate bicycle presentation over $R' = \Ft[G'] \simeq \Ft[x, y]/I$, where $I = (x^{15} - 1, y^3 - 1)$, with $a = x^9 + y + y^2$ and $b = 1 + x^2 + x^7$~\cite{bravyi_highthreshold_2024}.

\subsubsection{Scalar surgery}
We first analyse the presentation over $R'$.
Since $\ell = 15$ and $m = 3$ are odd, its primary components coincide with their associated primes, $T_i = P_i$, and every local block is a field.
Let $J = (a, b)$.
In lexicographic order $x > y$, a Gr\"obner basis of $I + J$ is
\begin{align*}
    \cG = \{x^2 + x + 1, y^2 + y + 1\},
\end{align*}
with standard monomial basis $\mathfrak{b}_J = \{1, x, y, xy\}$, giving $k = 2 |\mathfrak{b}_J| = 8$.

By \Cref{thm:block_koszul_homology}, only two of the $14$ blocks support homology.
The factorisation
\begin{align*}
    b = (x^2 + x + 1)(x^5 + x^4 + x^2 + x + 1)
\end{align*}
gives the corresponding primary components
\begin{align*}
    T_1 &= (x + y, x^2 + x + 1), \\
    T_2 &= (x + y + 1, x^2 + x + 1),
\end{align*}
and both local factors satisfy $Q_i \simeq \mathbb{F}_{2^2}$.
Each homology block therefore contains four logical qubits.

Scalar surgery is all-or-nothing on these field factors.
Up to measurement equivalence and composition, two non-trivial representatives are
\begin{align*}
    c_1 &= x^4y^2 + x^5y^2 + x^7y,\\
    c_2 &= xy + x^3y^2 + x^{11}.
\end{align*}
They measure the blocks associated with $Q_2$ and $Q_1$, respectively, since for $\{i, j\} = \{1, 2\}$,
\begin{align*}
    c_i \equiv 0 \pmod{T_i + J}, \qquad c_i \equiv u_j \pmod{T_j + J},
\end{align*}
for some unit $u_j$.
Each surgery measures four logical operators across both sectors of the corresponding homology block.

These representatives are measurement-equivalent to the lower-weight elements
\begin{align*}
    c_1' &= x + y,\\
    c_2' &= x + y + 1.
\end{align*}
Both $c_1$ and $c_2$ yield merged codes with
\begin{align*}
    \llbracket 135, 4, \leq 8 \rrbracket, \qquad (\omega_c, \omega_q) = (9, 9),
\end{align*}
whereas $c_1'$ and $c_2'$ give merged-code distances at most $6$.

\subsubsection{Completely fine measurements}\label{sec:ls-constructions-r35-rank1}

We now return to the presentation over $R = \Ft[x]/(x^5 - 1)$, whose block decomposition is
\begin{align*}
    R_1 &= \Ft[x]/(x - 1) \cong \Ft,\\
    R_2 &= \Ft[x]/(x^4 + x^3 + x^2 + x + 1) \cong \F_{16}.
\end{align*}
All logical operators are supported on $R_1$, so $H_1 = H_{1,(1)}$ is an eight-dimensional $\Ft$-vector space on which $x$ acts trivially.
Because $R_1 \cong \Ft$, its submodules are precisely the binary subspaces of $H_1$; hence this block is \emph{completely fine} by~\Cref{def:fine}.
Thus every logical subspace is addressable, and the remaining task is to find sparse chain-map representatives.
We describe this search in \Cref{sec:ls-constructions-complete-add}.

\subsubsection{Chain pair measurements}\label{sec:ls-constructions-radial35-chainpair}
As a lifted product code, it also admits \emph{lifted product surgery}; see \Cref{sec:lifted-prod-surgery}. Let $\partial_A  = a$ be the $3 \times 3$ protograph defined at the beginning of this section. Take a chain pair $\alpha \in \mathcal{E}_A$ on the first factor and the identity $(\idty, \idty) \in \mathcal{E}_B$ on the second. 
The resulting chain map has components
\begin{align*}
    f_2 = \alpha_1 \otimes \idty, \quad
    f_1 = (\alpha_0 \otimes \idty,\ \alpha_1 \otimes \idty), \quad
    f_0 = \alpha_0 \otimes \idty,
\end{align*}
and both surgery conditions reduce to $\partial_A \alpha_1 = \alpha_0 \partial_A$.

For the rank-one construction of \Cref{sec:lifted-prod-surgery}, take $v = w = (1, 0, 1)^{\mathsf T}$.
Then
\begin{align*}
    \partial_A v = \lambda v, \qquad w^{\mathsf T} \partial_A = \lambda w^{\mathsf T},
    \qquad \lambda = x^2 + x^4,
\end{align*}
so $\alpha_0 = \alpha_1 = v w^{\mathsf T}$ is the chain pair
\begin{align*}
    \alpha_i = \begin{pmatrix} 1 & 0 & 1\\ 0 & 0 & 0\\ 1 & 0 & 1 \end{pmatrix}.
\end{align*}
The homology is supported on $R_1 \simeq \Ft$, where both protographs reduce to the all-ones matrix $J_3$.
Since $v \in \ker J_3$ and $w^{\mathsf T}$ is nonzero on $\ker J_3$, $\alpha_i$ has rank one on each factor.
The induced map therefore has rank $1 \cdot 2 + 1 \cdot 2 = 4$ and measures
\begin{align*}
    W = \operatorname{span}(Z_1 Z_7,\, Z_3,\, Z_4,\, Z_8), \qquad \dim_{\Ft} W = 4 \, .
\end{align*}
The merged code is $\llbracket 135, 4, \leq 10 \rrbracket$ with $(\omega_c, \omega_q) = (8, 8)$ and maximum $Z$-check weight $6$.

\subsection{$\llbracket 144,12,12 \rrbracket$ ``Gross Code''}\label{sec:ls-constructions-bb}
The Gross code~\cite{bravyi_highthreshold_2024} is a bivariate bicycle code over $R = B/I$, where $B = \Ft[x, y]$ and $I = (x^{12} - 1, y^6 - 1)$.
Its defining polynomials are $a = x^3 + y + y^2$ and $b = y^3 + x + x^2$.
Equivalently, its code complex is the Koszul complex $K_\bullet(a, b)$ over $R$.

\subsubsection{Scalar surgery}
Let $J = (a, b)$.
In lexicographic order $x > y$, a Gr\"obner basis of $I + J$ is
\begin{equation*}
    \cG = \{x^2 + x + y^3, y^4 + y^2 + 1, xy^2 + xy + x + y^3 + 1\},
\end{equation*}
with standard monomial basis
\begin{align*}
    \mathfrak{b}_J = \{1, x, xy, y, y^2, y^3\}.
\end{align*}
Thus $\dim_{\Ft} R/J = 6$, and \Cref{prop:dim_koszul_two} gives $k = 12$; see also~\cite{postema_existence_2026}.

By \Cref{thm:block_koszul_homology}, a local block $(R_i, \mm_i) = (B/T_i, P_i/T_i)$ supports homology if and only if $J \subseteq P_i$.
Exactly two blocks satisfy this criterion:
\begin{align*}
    T_1 &= \bigl((y^2 + y + 1)^2, (x + y^2 + 1)^4\bigr),
    & Q_1 &\simeq \mathbb{F}_{2^2},\\
    T_2 &= \bigl((y^2 + y + 1)^2, (x + y^2)^4\bigr),
    & Q_2 &\simeq \mathbb{F}_{2^2}[t]/(t^2).
\end{align*}
The homology blocks $H_{1,(1)}$ and $H_{1,(2)}$ have dimensions $4$ and $8$, respectively.
Up to measurement equivalence and composition, there are three non-zero blockwise scalar classes.

On the field factor $Q_1$, the representative
\begin{align*}
    c_1 = 1 + x^2y^2
\end{align*}
measures the four logical operators in $H_{1,(1)}$.

The two non-zero classes on the chain ring $Q_2$ follow its radical filtration
\begin{align*}
    Q_2 \supsetneq \mm_2 Q_2 \supsetneq 0, \qquad \mm_2 = (t).
\end{align*}
The representatives
\begin{align*}
    c_2 &= 1 + xy^2,\\
    c_2' &= 1 + xy + x^3y^2 + x^5y^4
\end{align*}
measure $H_{1,(2)}$ and $\mm_2 H_{1,(2)}$, containing eight and four logical operators, respectively.
This illustrates how the structure of a block that is not a field enables scalar surgery to select nested logical subspaces within it.
Combining the blockwise actions gives
\begin{align*}
    c_3 = 1 + xy + x^2y^3 + x^6y,
\end{align*}
which measures the eight-dimensional subspace
\begin{align} \label{eq:gross-subspace}
    W = H_{1,(1)} \oplus \mm_2 H_{1,(2)}.
\end{align}
This subspace is of particular interest because, using the alternative group algebra presentation introduced below, our numerical search finds a lifted surgery that measures any individual logical $Z$ operator within $W$.

\subsubsection{Lifted surgery}
The Gross code also has a lifted product presentation over $R' = \Ft[t, z]/(t^4 - 1, z^2 - 1)$.
Its protographs are
\begin{align*}
    A = \begin{pmatrix}
        t & z & z\\
        1 & t & z\\
        1 & 1 & t
    \end{pmatrix}, \quad B = \begin{pmatrix}
        z & t & t\\
        1 & z & t\\
        1 & 1 & z
    \end{pmatrix}.
\end{align*}
Over $\Ft$,
\begin{align*}
    t^4 - 1 = (t - 1)^4, \qquad z^2 - 1 = (z - 1)^2.
\end{align*}
Hence $R'$ is local, with maximal ideal $\mm = (1 + t, 1 + z)$, and has a single block.
The lifted-product support criterion states that a block supports homology if and only if both determinants are non-units~\cite{panteleev_quantum_2021}.
Here
\begin{align*}
    \det A &= t^3 + tz + z + z^2,\\
    \det B &= z^3 + tz + t + t^2.
\end{align*}
Both determinants lie in $\mm$, so the unique block supports homology.

Applying \Cref{alg:lifted_surgery} to the two presentations gives
\begin{align*}
    \dim_{\Ft} \mathcal{M}_{R} = 14, \qquad \dim_{\Ft} \mathcal{M}_{R'} = 46,
\end{align*}
with $\mathcal{M}_{R} \subseteq \mathcal{M}_{R'}$.
Our numerical search over $R'$ finds a lifted surgery for every rank-$1$ $\Ft$-subspace of $W$ (\cref{eq:gross-subspace}).
Thus each logical operator in this eight-dimensional subspace is individually addressable.
The lifted product presentation also enables structured constructions such as chain-pair surgery.

\subsection{$\llbracket 210,10,10 \rrbracket$ code}\label{sec:ls-constructions-gb}
The $\llbracket 210,10,10 \rrbracket$ generalised bicycle code is defined over $R = \Ft[x]/(x^{105} - 1)$ by $a = 1 + x + x^{17}$ and $b = 1 + x^4 + x^5$.
Equivalently, its code complex is the Koszul complex $K_\bullet(a, b)$ over $R$.

For univariate codes, the primary decomposition and Gr\"obner basis computation in \Cref{alg:koszul_poly} reduce to polynomial arithmetic.
Let $\ell = 105$, $I = (x^\ell - 1)$, and $J = (a, b)$.
A Gr\"obner basis of $I + J$ is generated by
\begin{align*}
    g &= \gcd(a, b, x^\ell - 1) \\
      &= 1 + x^4 + x^5 \\
      &= (1 + x + x^2)(1 + x + x^3) = f_1 f_2.
\end{align*}
The irreducible factors define the supported primary components and local factors
\begin{align*}
    T_1 = (f_1), &\qquad Q_1 \simeq \mathbb{F}_{2^2},\\
    T_2 = (f_2), &\qquad Q_2 \simeq \mathbb{F}_{2^3}.
\end{align*}
The homology blocks $H_{1,(1)}$ and $H_{1,(2)}$ contain four and six logical operators, respectively.
Scalar surgery is all-or-nothing on these field factors.
Representatives of the two non-zero classes are
\begin{align*}
    c_1 &= 1 + x + x^{29}, \\
    c_2 &= 1 + x^3 + x^{29}.
\end{align*}
The representative $c_1$ measures $H_{1,(2)}$, while $c_2$ measures $H_{1,(1)}$.

\subsection{Computer search over fine blocks}\label{sec:ls-constructions-complete-add}

Here we describe the search used to find optimised lifted surgery maps for several codes with completely fine blocks.
Since the search space for lifted surgery maps is very large, it is helpful to find ways to reduce the number of maps we need to construct.

To do this, we will consider the automorphisms of the codes.
An automorphism splits the logical subspaces of a code into orbits.
Once we find lifted surgery maps for a logical subspace, we can translate these into maps for the entire orbit by applying the automorphism.
We can thus construct lifted surgery maps for a smaller set of \emph{seed} logical subspaces, and use these to measure a larger set of logical subspaces~\cite{webster_explicit_2025}.

For each code, we numerically determine the automorphism group.
For each code, the automorphism group splits into two parts, $G \times H$.
Since any lifted surgery map measures the entire orbit of $G$, we need only consider the second part, $H$.
Thus when performing the search, we only need to optimise one gadget per orbit under $H$.

Recall from \Cref{sec:surgery-equivalence} that homotopic lifted surgery maps have the same action, and the set of homotopy-equivalent maps forms an affine subspace of the space of lifted surgery maps.
We search through these subspaces using simulated annealing to optimise the maps.
For each subspace, we minimise the maximum degree of the merged code, given by the maximum of every check weight and every qubit degree, and we only keep maps such that the distance of the merged code is at least as large as the distance of the base code.
We upper-bound the distances of the merged codes using an information-set-decoding-based algorithm~\cite{pryadko_qdistrnd_2022}.

We also perform a targeted search for lifted product maps (cf.~\Cref{sec:lifted-prod-surgery}) for codes that can be described by non-trivial products.
We find that lifted product maps allow us to measure logical subspaces across a variety of ranks.
Furthermore, these maps are able to attain maximum degrees smaller than those of the more general lifted surgery maps.
For example, we find rank-4 measurements of maximum degree 8 for the $\llbracket 90, 8, 10 \rrbracket$ code and maximum degree 9 for the $\llbracket 198, 8, 16 \rrbracket$ code.

\section{Cross-code measurements with lifted bridges}\label{sec:bridging}
\begin{figure}[t]
    \centering
\includegraphics[width=0.95\linewidth]{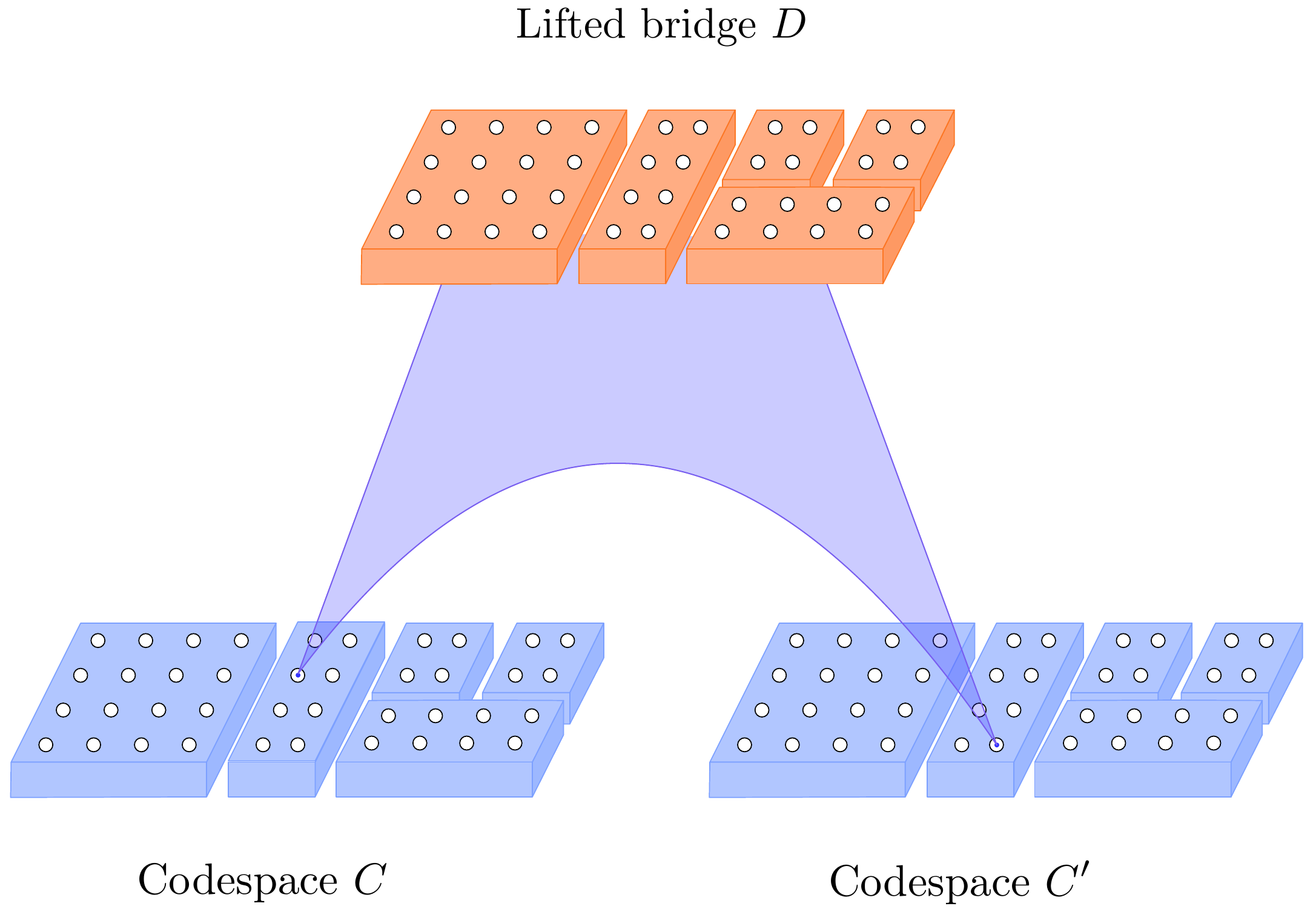}
    \caption{
        Graphical representation of bridged lifted surgery.
        On the bottom row are two copies of the same base codespace, $C$ and $C'$, over the group algebra $R$.
        On the top row is an auxiliary copy of the same codespace, $D$.
        The purple sheet connects the auxiliary code $D$ to the two codespaces $C$ and $C'$ and performs a measurement of a \emph{joint} logical operator across the two copies of the base code. 
    }
    \label{fig:bridge}
\end{figure}
In the previous sections, we focused on measuring logical operators in a single base code.
In general, it is important that we can measure operators \emph{jointly} across different base codes.
For instance, in an algorithm we might want to measure $Z_1 \otimes Z_2$, where $Z_1$ and $Z_2$ come from distinct codes.
Standard QLDPC surgery achieves this by \emph{bridging} the auxiliary systems attached to each base code~\cite{cross_improved_2025,swaroop_universal_2024}.
Ref.~\cite{cowtan_fast_2025} introduced \emph{partial block reading}, which jointly measures logical operators across base codes, $C_\bullet$ and $C'_\bullet$, by introducing an auxiliary complex, $D_\bullet$.
When $D_\bullet$ is a subcode of both base codes, this allows us to measure all logical operators of $D_\bullet$ jointly between $C_\bullet$ and $C'_\bullet$.

We introduce a generalisation of this scheme based on our formalism; see~\Cref{fig:bridge}.
We describe a joint logical measurement via the following diagram:
\begin{equation*}
\begin{tikzcd}
	&& {D_\bullet} && \\
	{C_\bullet} &&&& {C'_\bullet}
	\arrow["{f_\bullet}"', from=1-3, to=2-1]
	\arrow["{g_\bullet}", from=1-3, to=2-5]
\end{tikzcd}
\end{equation*}
where the complexes are over $R=\Ft[G]$ and the morphisms defining the chain maps $f_\bullet, g_\bullet$ are $R$-linear.
\begin{definition}[Bridged lifted surgery]\label{def:bridged_surgery}
    Let $C_\bullet, C'_\bullet$ be group algebra codes over $R$ and let $D_\bullet$ be an auxiliary complex.
    A \emph{bridged lifted surgery map} is an $R$-linear chain map
    \begin{align*}
        m_\bullet = (f_\bullet, g_\bullet)\colon D_\bullet & \to C_\bullet \oplus C'_\bullet \\
        d & \mapsto (f_\bullet(d), g_\bullet(d)) \, ,
    \end{align*}
    where $f_\bullet \colon D_\bullet \to C_\bullet, g_\bullet \colon D_\bullet \to C'_\bullet$ are chain maps.
    The merged code corresponds to the code $\cone(m_\bullet)$.
\end{definition}
Note that $m_\bullet$ is a chain map iff $f_\bullet$ and $g_\bullet$ are, that is, 
\begin{align*}
    m_\bullet \in \ls_R(D_\bullet,C_\bullet \oplus C'_\bullet)
\end{align*}
is a lifted surgery map on the direct sum complex. 
\begin{restatable}[Bridged lifted surgery addressable modules]{proposition}{propOneBridging}
\label{prop:bridging_measured_set}
    Let $m_\bullet = (f_\bullet, g_\bullet)$ be a bridged lifted surgery map.
    The set of measured logical operators is the $R$-submodule
    \begin{align*}
        \im (m_*)
        &= \bigl\{\, (f_*[d],\, g_*[d]) \;\big|\; [d] \in H_1(D_\bullet) \,\bigr\} \\
        &\quad\subseteq H_1(C_\bullet) \oplus H_1(C'_\bullet)
    \end{align*}
    where $m_*$ is the induced map in homology.
    That is, $\cone(m_\bullet)$ measures 
    \begin{align*}
        \rank_{\Ft} m_*
    \end{align*}
     independent
    logical operators. 
    Moreover, the number $\dim_{\Ft} H_1(\cone(m_\bullet))$ of logical operators of the merged code is
    \begin{align*}
        \dim_{\Ft} H_1(C_\bullet \oplus C'_\bullet) - \rank_{\Ft} m_* + \dim_{\Ft} \ker m_*^0 ,
    \end{align*}
    where $m_*^0$ denotes the map induced by $m_\bullet$ on $H_0$.
\end{restatable}
\begin{proof}    
    See \Cref{sec:app-proofs-bridging}.
\end{proof}

To measure \emph{joint} logical operators, as opposed to parallel measurement of two observables across different base codes, the chain maps need to fulfil a compatibility condition that the following proposition characterises.
We give the proof in~\Cref{sec:app-proofs-bridging}.
\begin{restatable}[Joint logical measurements]{proposition}{propTwoBridging}
\label{prop:bridging_joint}
    Let $m_\bullet = (f_\bullet, g_\bullet)$ be a bridged lifted surgery map.
    Every logical operator is measured \emph{jointly} if $f_\bullet$ and $g_\bullet$ have the same kernel $K = \ker f_* = \ker g_*$ on the homology.
    In this case, $\cone(m_\bullet)$ measures
    \begin{align*}
    r &= \dim_{\Ft} H_1(D_\bullet) - \dim_{\Ft} \ker f_*\\
    &= \dim_{\Ft} H_1(D_\bullet) - \dim_{\Ft} \ker g_*
    \end{align*}
    independent observables. 
    The measured space is 
    \begin{align*}
        \im m_* \simeq H_1(D_\bullet)/K \, .
    \end{align*}
    Moreover, let $d_1,\dots,d_r\in\ker\partial_1^D$ be cycles such that $[d_1]+K,\dots,[d_r]+K$ form an $\Ft$-basis of $H_1(D_\bullet)/K$. 
    Then
    \begin{align*}
        m_*[d_j] = \bigl([f_1(d_j)],[g_1(d_j)]\bigr), \qquad j=1,\dots,r,    
    \end{align*}
    form a basis of the measured space $\im m_*$. 
    The corresponding independent measured logical operators are
    \begin{align*}
        f_1(d_j) \otimes g_1(d_j), \qquad j=1,\dots,r.    
    \end{align*}
\end{restatable}
Note that this also means that we can measure joint logical operators that do not have to be the same, i.e., $Z_i \otimes Z_j, i \neq j$ with \emph{different} maps $f_\bullet \neq g_\bullet$ as long as $\ker f_* = \ker g_*$.

As in~\Cref{sec:ls-constructions-complete-add}, we can perform a search for optimised lifted bridges.
We summarise some results for the $\llbracket 90, 8, 10\rrbracket$ and $\llbracket 198, 8, 16 \rrbracket$ codes in~\Cref{tab:bridge-results}.
For both of these codes, every $8\times 8$ matrix over $\Ft$ is the induced action of some $R$-linear chain map.
Thus we can find chain maps $f_\bullet, g_\bullet$ to measure any pair $Z_i \otimes Z_j$ for arbitrary $i,j \in [8]$.
Moreover, this also holds for any $r$-tuples of logical operators for these codes. 
Note that this flexibility is due to the fact that all logical operators for these codes are supported on a single completely fine block.
\begin{table}[tb]
    \centering
    \small
    \renewcommand{\arraystretch}{1.25}
    \setlength{\tabcolsep}{2pt}
    \begin{tabular*}{\columnwidth}{@{\extracolsep{\fill}}cccccc@{}}
        \toprule
        base code & rank & orbit size & merged code & $\omega_c$ & $\omega_q$ \\
        \midrule
        $\llbracket 90, 8, 10\rrbracket$
        & $1$
        & $9$
        & $\llbracket 225, 15, \leq 10\rrbracket$
        & $13$
        & $13$ \\
        & $4$
        & $3$
        & $\llbracket 225, 12, \leq 10\rrbracket$
        & $8$
        & $9$ \\
        \midrule
        $\llbracket 198, 8, 16\rrbracket$
        & $4$
        & $1$
        & $\llbracket 495, 12, \leq 16\rrbracket$
        & $8$
        & $9$ \\
        \bottomrule
    \end{tabular*}
    \caption{
        Lifted bridge construction data for radial codes.
        Each construction uses two copies of the base code and a third copy acts as the bridge.
        The orbit size counts the logical subspaces of the given rank measurable with equivalent merged codes. 
        Here $\omega_c$ and $\omega_q$ denote the maximum check weight and qubit degree of the merged code.
    }
    \label{tab:bridge-results}
\end{table}

\section{Numerical simulations}\label{sec:numerics}
\begin{figure*}[t]
    \centering
    \subfloat[
     Block error rates for lifted scalar surgeries with the $\llbracket 210, 10, 10 \rrbracket$ GB code and $\llbracket 144, 12, 12 \rrbracket$ Gross code.
    ]{           
        \includegraphics[width=0.49\linewidth]{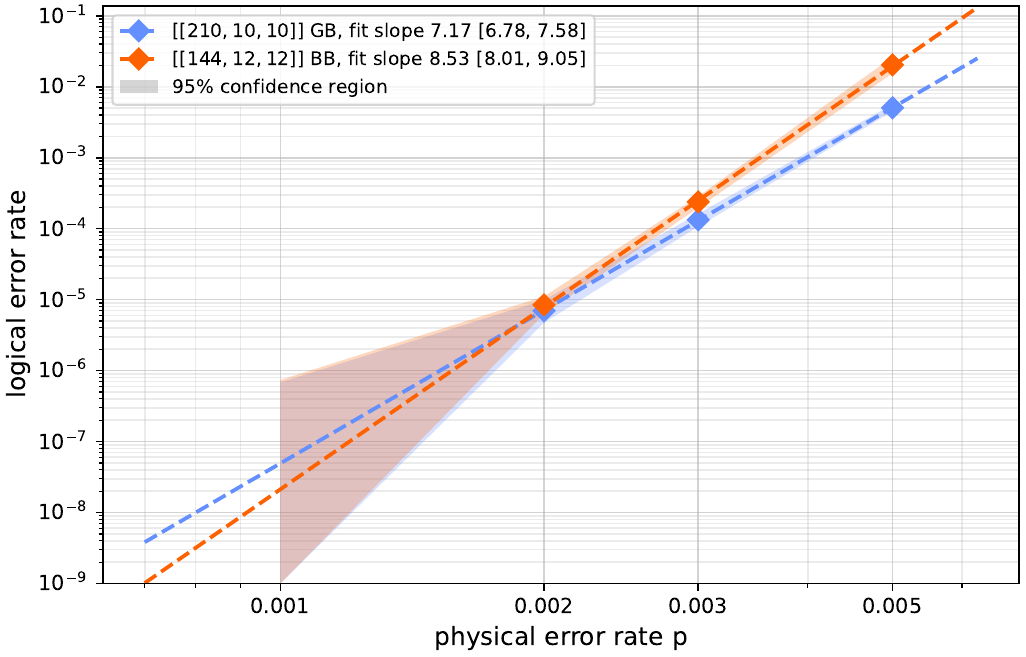}\label{fig:sim-results-a}
    }
    ~
    \subfloat[
        Logical error rates for various surgery operations at $p=10^{-3}$ for the completely fine addressable quantum radial codes.
    ]{
        \includegraphics[width=0.49\linewidth]{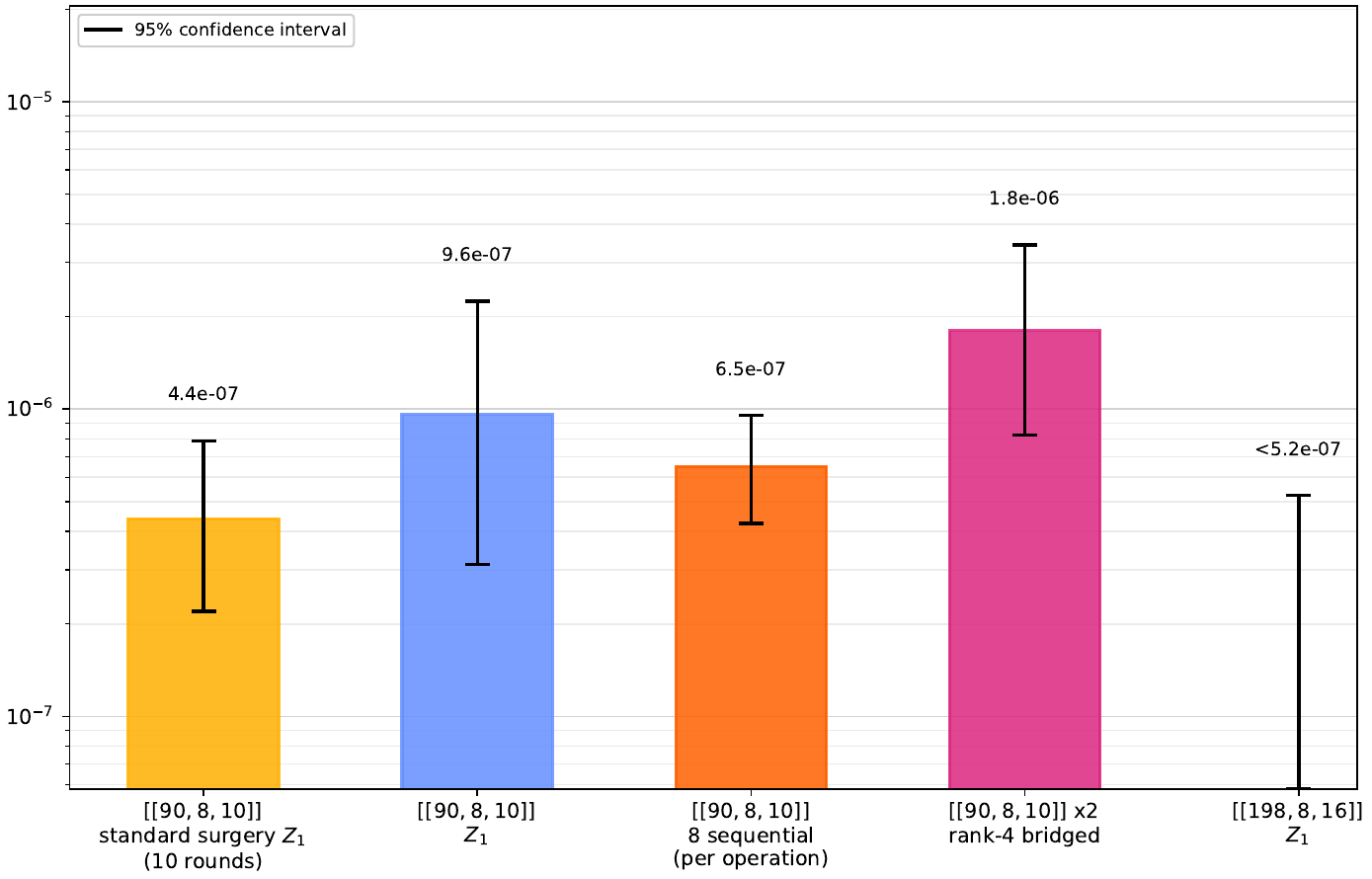}\label{fig:sim-results-b}
    }
    \caption{
        Numerical decoding simulation results.
        (a) Logical error rates for a single round of lifted (scalar) surgery with the $\llbracket 210, 10, 10 \rrbracket$ GB and the Gross code.
        Dashed lines are maximum-likelihood fits and the shaded region corresponds to a $95$\% confidence region.
        At $p=10^{-3}$ we did not find a logical error for either of the codes in at least $5\times 10^{6}$ shots.
        (b) Block error rates per logical operation at $p=10^{-3}$ for various numerical surgery experiments on the $\llbracket 90, 8, 10 \rrbracket$ and $\llbracket 198, 8, 16 \rrbracket$ radial codes.
        The first bar (yellow) shows the LER of a standard surgery experiment of $Z_1$ on the $\llbracket 90, 8, 10 \rrbracket$ code with $10$ rounds of syndrome extraction, all decoded at once;
        the second bar (blue) shows a rank-1 lifted surgery measurement of $Z_1$ on the $\llbracket 90, 8, 10 \rrbracket$ code;
        the third bar (orange) shows the LER per operation for 8 sequential rank-1 measurements on the $\llbracket 90, 8, 10 \rrbracket$ code, with no intervening rounds of syndrome extraction;
        the fourth bar (magenta) shows the LER of a rank-4 bridged measurement on two copies of the $\llbracket 90, 8, 10 \rrbracket$ code;
        the rightmost point depicts the upper bound on the LER for a $Z_1$ measurement on the $\llbracket 198, 8, 16 \rrbracket$ code. 
        We found no errors in over $7\times 10^{6}$ shots.
    }
    \label{fig:sim-results}
\end{figure*}
In this section, we present the results of numerical simulations of lifted surgery under circuit-level depolarising noise.

\subsection{Syndrome extraction circuits}
To efficiently design good syndrome extraction schedules, we exploit the group algebra structure of the merged codes.
When performing scalar surgery on a Koszul code, the resulting merged code is also a Koszul code.
For example, when performing scalar surgery on a generalised bicycle code, the merged code can be viewed as a tricycle code~\cite{jacob_singleshot_2026,menon_magic_2026}.
This allows us to use the already optimal schedules for tricycle codes.
When the merged code cannot be described as a Koszul code, we instead use integer linear programming to find a valid schedule~\cite{vittal_flagproxy_2024}.
A direct formulation over all physical qubits would be prohibitively large, so we instead formulate the problem over the group algebra and require the schedule to respect the group action.
It is therefore sufficient to optimise the circuit for a set of orbit representatives, since translating the resulting schedule according to the group action generates the full syndrome extraction circuit.
This symmetry reduction substantially decreases the number of variables and constraints, making the search feasible.

We focus on logical $Z$ measurements and refer to this as the primal basis.
We use a standard circuit-level depolarising noise model where preparation, gates, idling, and readouts are noisy and parameterised by a physical parameter $p$.

We construct the simulated circuits as follows:
\begin{enumerate}
    \item Initialise the data qubits of the base code in the primal basis, and the data qubits of the auxiliary system in the dual basis. 
    \item Measure the checks of the merged code once.
    \item Read out the data qubits of the base code in the primal basis, and the data qubits of the auxiliary system in the dual basis.
    \item Infer the logical measurement outcome from the primal checks in the auxiliary system.
\end{enumerate}
The defined detectors include the meta-checks and we infer the dual checks, indexed by elements of $F$ in~\cref{eq:fast-surgery} from the final readout of the data qubits in the auxiliary system.
We use the final data qubit readouts in the base code to reconstruct primal checks and define detectors that compare their values with those from the merged code measurement.

\subsection{Decoding algorithm}
When constructing lifted surgery gadgets, we optimise the sparsity of stabiliser checks that we physically measure, as these directly influence the depth of the syndrome extraction circuit.
We do not optimise the sparsity of the meta-checks or the stabiliser checks that we infer during readout of the auxiliary system.
This means the decoding matrix we obtain can be relatively dense.
This is of practical concern since the runtimes of many decoding algorithms are heavily dependent on the sparsity of the decoding matrix.

We implement several heuristic optimisations to speed up decoding for dense decoding matrices.
Instead of performing decoding on all the detectors at once, we first decode using detectors from the dual sector only.
We then use the correction from this first pass to update the priors for each error mechanism, before decoding the detectors from the primal sector; see, e.g., Appendix D of Ref.~\cite{roffe_biastailored_2023}.
Instead of decoding directly using the circuit's detectors, we first sparsify the decoding matrix as mentioned above.
In particular, we construct a sparse basis for the rowspace of the decoding matrix by conducting a randomised search for low-weight rowspace elements.

For each decoding pass, we use a most-likely error (MLE) decoder, which we implement using integer linear programming (ILP).
Before running the MLE decoder on each sector, we run a pass of BP+LSD~\cite{hillmann_localized_2025} (LSD-CS order 0 with unoptimised default parameters) to warm-start the ILP solver with a valid solution.
The decoder falls back to this solution if the ILP solver fails to find a lower-weight correction before timing out.
We set a $2$-hour timeout for all simulations reported here.

\subsection{Results}
The results of our numerical decoding simulations are given in~\Cref{fig:sim-results}.
We start by benchmarking the logical performance of an isolated scalar surgery operation on the $\llbracket 210,10,10 \rrbracket$ GB code, the Gross code, the $\llbracket 90, 8, 10 \rrbracket$ radial code, and the $\llbracket 198,8,16\rrbracket$ radial code (cf.~\Cref{sec:app-r311-code}).

\Cref{fig:sim-results-a} shows the logical error rate of a single round of scalar surgery, measuring $6$ logical qubits for the $\llbracket 210,10,10 \rrbracket$ GB code, and $8$ logical qubits for the Gross code.

For the radial codes, we simulated measuring a single logical operator at $p=10^{-3}$, and the bar plot in~\Cref{fig:sim-results-b} shows the achieved logical error rates.
For each of the above cases, we obtained logical failure rates below $10^{-6}$ at $p=10^{-3}$.

We also simulated a $10$-round experiment using standard code surgery~\cite{williamson_lowoverhead_2026,ide_faulttolerant_2025} on the $\llbracket 90, 8, 10\rrbracket$ code, and found that performance was only marginally better than that of lifted surgery, showing that lifted surgery can match the performance of standard code surgery while allowing for an order-of-magnitude decrease in the time overhead.

We also investigated performing several lifted surgery operations in sequence, with no rounds of memory between subsequent operations.
This is the scenario analysed in Ref.~\cite{cowtan_fast_2025}.
We perform $8$ sequential surgery measurements on the $\llbracket 90, 8, 10 \rrbracket$ code and decode the entire block of operations. 
The third bar from the left (orange) in~\Cref{fig:sim-results-b} shows the logical failure rate at $p = 10^{-3}$.
We can see that the logical failure rate per logical operation matches the logical failure rate for a single isolated operation.

Finally, we also investigated the logical error rate for a joint logical measurement across two base codes with bridged lifted surgery. 
We take two copies of the $\llbracket 90, 8, 10 \rrbracket$ code as base codes and simulate one joint measurement of a rank-4 lifted surgery operation.
As can be seen from the fourth bar from the left (magenta) in~\Cref{fig:sim-results-b}, we find that the logical failure rate is again comparable to that of a single lifted surgery operation on one code for the constructed maps.

\section{Discussion}\label{sec:discussion}
Our work gives an explicit construction for fast surgery on Abelian group algebra codes, a prominent family of QLDPC codes.
We give a rigorous characterisation of lifted surgery and construct several examples, showing that with lifted surgery we can obtain fast logical processing at utility scale.
Many recent architectures based on LDPC codes use Abelian group algebra codes~\cite{webster_pinnacle_2026, cain_shor_2026, yoder_tour_2025}, and thus our scheme has the potential to directly improve the runtime of algorithms on these architectures.

Realising these gains will require integrating lifted surgery into a fault-tolerant architecture.
Since fast surgery only allows for CSS-type measurements, we will need to develop compilation strategies that minimise the use of non-CSS measurements and integrate methods for performing non-CSS operations such as standard surgery, automorphism gates, or $Y$-state distillation.
We will also want to develop compilation strategies that make optimal use of the parallel measurements available for certain codes.

We also want to find optimised hardware layouts for performing lifted surgery operations.
Lifted surgery is particularly suitable for hardware platforms that allow for dynamic reconfiguration, such as neutral atoms, trapped ions, and spin qubits.
Lifted surgery preserves the underlying group symmetry of the code, which we can exploit to design syndrome extraction circuits with low reconfiguration depths.
Conducting in-depth numerical simulations for hardware-informed models will be important for benchmarking the practical performance of lifted surgery.

\section{Acknowledgements}
We thank all our colleagues at Iceberg Quantum for helpful discussions and feedback on the manuscript.

\bibliographystyle{unsrt}
\bibliography{references}

\begin{thebibliography}{10}

\bibitem{gidney_how_2025}
Craig Gidney.
\newblock How to factor 2048 bit rsa integers with less than a million noisy
  qubits, 2025.

\bibitem{breuckmann_quantum_2021}
Nikolas~P. Breuckmann and Jens~Niklas Eberhardt.
\newblock Quantum low-density parity-check codes.
\newblock {\em PRX Quantum}, 2:040101, October 2021.

\bibitem{webster_explicit_2025}
Paul Webster, Samuel~C. Smith, and Lawrence~Z. Cohen.
\newblock Explicit construction of low-overhead gadgets for gates on quantum
  ldpc codes, 2025.

\bibitem{bravyi_highthreshold_2024}
Sergey Bravyi, Andrew~W. Cross, Jay~M. Gambetta, Dmitri Maslov, Patrick Rall,
  and Theodore~J. Yoder.
\newblock High-threshold and low-overhead fault-tolerant quantum memory.
\newblock {\em Nature}, 627(8005):778--782, March 2024.

\bibitem{bhardwaj_highrate_2026}
Aditya Bhardwaj, Muzhou Ma, Nadine Meister, Robbie King, Dolev Bluvstein, John
  Preskill, Madelyn Cain, Qian Xu, and Hsin-Yuan Huang.
\newblock High-rate {qLDPC} processors, 2026.
\newblock arXiv:2607.28795 [quant-ph].

\bibitem{webster_pinnacle_2026}
Paul Webster, Lucas Berent, Omprakash Chandra, Evan~T. Hockings, Nouédyn
  Baspin, Felix Thomsen, Samuel~C. Smith, and Lawrence~Z. Cohen.
\newblock The pinnacle architecture: Reducing the cost of breaking rsa-2048 to
  100 000 physical qubits using quantum ldpc codes, 2026.

\bibitem{yoder_tour_2025}
Theodore~J. Yoder, Eddie Schoute, Patrick Rall, Emily Pritchett, Jay~M.
  Gambetta, Andrew~W. Cross, Malcolm Carroll, and Michael~E. Beverland.
\newblock Tour de gross: A modular quantum computer based on bivariate bicycle
  codes, 2025.

\bibitem{cain_shor_2026}
Madelyn Cain, Qian Xu, Robbie King, Lewis R.~B. Picard, Harry Levine, Manuel
  Endres, John Preskill, Hsin-Yuan Huang, and Dolev Bluvstein.
\newblock Shor's algorithm is possible with as few as 10,000 reconfigurable
  atomic qubits, 2026.

\bibitem{litinski_game_2019}
Daniel Litinski.
\newblock A {G}ame of {S}urface {C}odes: {L}arge-{S}cale {Q}uantum {C}omputing
  with {L}attice {S}urgery.
\newblock {\em {Quantum}}, 3:128, March 2019.

\bibitem{he_extractors_2025}
Zhiyang He, Alexander Cowtan, Dominic~J. Williamson, and Theodore~J. Yoder.
\newblock Extractors: Qldpc architectures for efficient pauli-based
  computation, 2025.

\bibitem{cohen_lowoverhead_2022}
Lawrence~Z. Cohen, Isaac~H. Kim, Stephen~D. Bartlett, and Benjamin~J. Brown.
\newblock Low-overhead fault-tolerant quantum computing using long-range
  connectivity.
\newblock {\em Science Advances}, 8(20):eabn1717, 2022.

\bibitem{ide_faulttolerant_2025}
Benjamin Ide, Manoj~G. Gowda, Priya~J. Nadkarni, and Guillaume Dauphinais.
\newblock Fault-tolerant logical measurements via homological measurement.
\newblock {\em Phys. Rev. X}, 15:021088, Jun 2025.

\bibitem{williamson_lowoverhead_2026}
Dominic~J. Williamson and Theodore~J. Yoder.
\newblock Low-overhead fault-tolerant quantum computation by gauging logical
  operators.
\newblock {\em Nature Physics}, 22(4):598–603, Apr 2026.

\bibitem{baspin_fast_2025}
Nouédyn Baspin, Lucas Berent, and Lawrence~Z. Cohen.
\newblock Fast surgery for quantum ldpc codes, 2025.

\bibitem{cowtan_fast_2025}
Alexander Cowtan, Zhiyang He, Dominic~J. Williamson, and Theodore~J. Yoder.
\newblock Fast and fault-tolerant logical measurements: Auxiliary hypergraphs
  and transversal surgery, 2025.

\bibitem{chang_constanttime_2026}
Kathleen Chang, Zhiyang He, Theodore~J Yoder, Guanyu Zhu, and Tomas
  Jochym-O'Connor.
\newblock Constant-time surgery on 2d hypergraph product codes with
  near-constant space overhead.
\newblock {\em arXiv preprint arXiv:2603.02157}, 2026.

\bibitem{tripier_faulttolerant_2026}
Felix Tripier, Woo~Chang Chung, Jacob Young, Safwan Alam, Bryce Bjork, Aharon
  Brodutch, Finn~Lasse Buessen, Nolan~J. Coble, Thomas Dellaert, Dmitri Maslov,
  Martin Roetteler, Edwin Tham, Mark Webster, Min Ye, John Gamble, Andrii
  Maksymov, J.~P. Marceaux, and Nicolas Delfosse.
\newblock Fault-tolerant quantum computing with trapped ions: The walking cat
  architecture, 2026.

\bibitem{scruby_highthreshold_2026}
Thomas~R. Scruby, Timo Hillmann, and Joschka Roffe.
\newblock High-threshold, low-overhead and single-shot decodable fault-tolerant
  quantum memory.
\newblock {\em PRX Quantum}, 7(2), Apr 2026.

\bibitem{vittal_flagproxy_2024}
Suhas Vittal, Ali Javadi-Abhari, Andrew~W. Cross, Lev~S. Bishop, and Moinuddin
  Qureshi.
\newblock Flag-proxy networks: Overcoming the architectural, scheduling and
  decoding obstacles of quantum ldpc codes.
\newblock In {\em 2024 57th IEEE/ACM International Symposium on
  Microarchitecture (MICRO)}, pages 718--734, 2024.

\bibitem{bombin_quantum_2009}
H.~Bombin and M.~A. Martin-Delgado.
\newblock Quantum {Measurements} and {Gates} by {Code} {Deformation}.
\newblock {\em Journal of Physics A: Mathematical and Theoretical},
  42(9):095302, March 2009.
\newblock arXiv:0704.2540 [quant-ph].

\bibitem{horsman_surface_2012}
Dominic Horsman, Austin~G Fowler, Simon Devitt, and Rodney~Van Meter.
\newblock Surface code quantum computing by lattice surgery.
\newblock {\em New Journal of Physics}, 14(12):123011, December 2012.

\bibitem{cross_improved_2025}
Andrew~W. Cross, Zhiyang He, Patrick~J. Rall, and Theodore~J. Yoder.
\newblock Improved qldpc surgery: Logical measurements and bridging codes,
  2025.

\bibitem{swaroop_universal_2024}
Esha Swaroop, Tomas Jochym-O'Connor, and Theodore~J. Yoder.
\newblock Universal adapters between quantum low-density parity check codes.
\newblock {\em PRX Quantum}, 7:010324, Feb 2026.

\bibitem{zhang_timeefficient_2025}
Guo Zhang and Ying Li.
\newblock Time-efficient logical operations on quantum low-density parity check
  codes.
\newblock {\em Physical Review Letters}, 134(7):070602, 2025.

\bibitem{cowtan_parallel_2026}
Alexander Cowtan, Zhiyang He, Dominic~J. Williamson, and Theodore~J. Yoder.
\newblock Parallel logical measurements via quantum code surgery.
\newblock {\em PRX Quantum}, 2026.

\bibitem{zheng_highrate_2025}
Guo Zheng, Liang Jiang, and Qian Xu.
\newblock High-{Rate} {Surgery}: towards constant-overhead logical operations,
  October 2025.
\newblock arXiv:2510.08523 [quant-ph].

\bibitem{zheng_logical_2026}
Han Zheng, Guo Zheng, Liang Jiang, and Qian Xu.
\newblock Logical computation with canonical lifted product codes, July 2026.
\newblock arXiv:2607.28605 [quant-ph].

\bibitem{blue_full_2026}
John Blue, Zhiyang He, Hengyun Zhou, and Isaac~L. Chuang.
\newblock Full extractors for logical processing in hypergraph product codes,
  2026.

\bibitem{zhou_resource_2025}
Hengyun Zhou, Casey Duckering, Chen Zhao, Dolev Bluvstein, Madelyn Cain,
  Aleksander Kubica, Sheng-Tao Wang, and Mikhail~D. Lukin.
\newblock Resource analysis of low-overhead transversal architectures for
  reconfigurable atom arrays.
\newblock In {\em Proceedings of the 52nd Annual International Symposium on
  Computer Architecture}, ISCA '25, page 1432–1448, New York, NY, USA, 2025.
  Association for Computing Machinery.

\bibitem{lubotzky_high_2018}
Alexander Lubotzky.
\newblock High dimensional expanders.
\newblock In {\em Proceedings of the international congress of mathematicians:
  Rio de Janeiro 2018}, pages 705--730. World Scientific, 2018.

\bibitem{oppenheim_new_2025}
Izhar Oppenheim and Inga Valentiner-Branth.
\newblock New cosystolic high-dimensional expanders from kms groups, 2025.

\bibitem{evra_bounded_2016}
Shai Evra and Tali Kaufman.
\newblock Bounded degree cosystolic expanders of every dimension.
\newblock In {\em Proceedings of the Forty-Eighth Annual ACM Symposium on
  Theory of Computing}, STOC '16, pages 36--48, New York, NY, USA, 2016.
  Association for Computing Machinery.

\bibitem{lin_quantum_2024}
Hsiang-Ku Lin and Leonid~P Pryadko.
\newblock Quantum two-block group algebra codes.
\newblock {\em Physical Review A}, 109(2):022407, 2024.

\bibitem{panteleev_degenerate_2021}
Pavel Panteleev and Gleb Kalachev.
\newblock Degenerate quantum ldpc codes with good finite length performance.
\newblock {\em Quantum}, 5:585, November 2021.

\bibitem{lin_abelian_2026}
Hsiang-Ku Lin, Pak~Kau Lim, Alexey~A. Kovalev, and Leonid~P. Pryadko.
\newblock Abelian multi-cycle codes for single-shot error correction.
\newblock {\em PRX Quantum}, 7:020326, May 2026.

\bibitem{gu_qgpu_2026}
Boren Gu, Andy~Zeyi Liu, Armanda~O Quintavalle, Qian Xu, Jens Eisert, and
  Joschka Roffe.
\newblock Qgpu: Parallel logic in quantum ldpc codes.
\newblock {\em arXiv preprint arXiv:2603.05398}, 2026.

\bibitem{breuckmann_logical_2025}
Nikolas~P. Breuckmann, Shin~Ho Choe, Jens~Niklas Eberhardt, Francisco
  Revson~Fernandes Pereira, and Vincent Steffan.
\newblock Logical operators and derived automorphisms of tile codes.
\newblock {\em Communications in Mathematical Physics}, 407(10):205, Aug 2026.

\bibitem{eberhardt_logical_2024}
Jens~Niklas Eberhardt and Vincent Steffan.
\newblock Logical operators and fold-transversal gates of bivariate bicycle
  codes.
\newblock {\em IEEE Trans. Inf. Theor.}, 71(2):1140–1152, February 2025.

\bibitem{chen_anyon_2025}
Keyang Chen, Yuanting Liu, Yiming Zhang, Zijian Liang, Yu-An Chen, Ke~Liu, and
  Hao Song.
\newblock Anyon theory and topological frustration of high-efficiency quantum
  low-density parity-check codes.
\newblock {\em Physical Review Letters}, 135(7):076603, 2025.

\bibitem{postema_existence_2026}
Jasper~J Postema and Servaas~JJMF Kokkelmans.
\newblock Existence and characterisation of bivariate bicycle codes.
\newblock {\em IEEE Transactions on Information Theory}, 2026.

\bibitem{lee_logical_2026}
Jong~Yeon Lee.
\newblock Logical spectroscopy: Lifted-product codes with addressable bases,
  2026.

\bibitem{sabo_spectral_2026}
Eric Sabo, Mahir Bilen~Can, and David Marquis.
\newblock Spectral theory of semisimple bivariate bicycle codes.
\newblock {\em arXiv preprint arxiv:2608.27565}, 2026.

\bibitem{atiyah_introduction_2018}
Michael~F Atiyah and Ian~Grant Macdonald.
\newblock {\em Introduction to commutative algebra}.
\newblock CRC press, 2018.

\bibitem{weibel_introduction_1994}
Charles~A. Weibel.
\newblock {\em An Introduction to Homological Algebra}.
\newblock Cambridge Studies in Advanced Mathematics. Cambridge University
  Press, 1994.

\bibitem{benhemou_automated_2026}
Asmae Benhemou and Noah Berthusen.
\newblock Automated logical {Clifford} gadgets for heterogeneous architectures
  via chain maps, July 2026.
\newblock arXiv:2607.02482 [quant-ph].

\bibitem{viszlai_matching_2025}
Joshua Viszlai, Willers Yang, Sophia~Fuhui Lin, Junyu Liu, Natalia Nottingham,
  Jonathan~M Baker, and Frederic~T Chong.
\newblock Matching generalized-bicycle codes to neutral atoms for low-overhead
  fault-tolerance.
\newblock In {\em 2025 IEEE International Conference on Quantum Computing and
  Engineering (QCE)}, volume~01, pages 688--699, 2025.

\bibitem{menon_magic_2026}
Varun Menon, J~Pablo Bonilla~Ataides, Rohan Mehta, Andi Gu, Daniel~Bochen Tan,
  and Mikhail~D Lukin.
\newblock Magic tricycles: Efficient magic-state generation with finite
  block-length quantum ldpc codes.
\newblock {\em Physical Review X}, 16(2):021014, 2026.

\bibitem{tillich_quantum_2014}
Jean-Pierre Tillich and Gilles Zémor.
\newblock Quantum ldpc codes with positive rate and minimum distance
  proportional to the square root of the blocklength.
\newblock {\em IEEE Transactions on Information Theory}, 60(2):1193--1202,
  2014.

\bibitem{jacob_singleshot_2026}
Abraham Jacob, Campbell McLauchlan, and Dan~E. Browne.
\newblock Single-{Shot} {Decoding} and {Fault}-tolerant {Gates} with
  {Trivariate} {Tricycle} {Codes}, January 2026.
\newblock arXiv:2508.08191 [quant-ph].

\bibitem{gianni_grobner_1988}
Patrizia Gianni, Barry Trager, and Gail Zacharias.
\newblock Gr{\"o}bner bases and primary decomposition of polynomial ideals.
\newblock {\em Journal of Symbolic Computation}, 6(2-3):149--167, 1988.

\bibitem{decker1999primary}
Wolfram Decker, Gert-Martin Greuel, and Gerhard Pfister.
\newblock Primary decomposition: algorithms and comparisons.
\newblock In {\em Algorithmic Algebra and Number Theory: Selected Papers From a
  Conference Held at the University of Heidelberg in October 1997}, pages
  187--220. Springer, 1999.

\bibitem{cox_ideals_2007}
David~A Cox, John~B Little, Donal O'shea, and John Little.
\newblock {\em Ideals, varieties, and algorithms: an introduction to
  computational algebraic geometry and commutative algebra}, volume~10.
\newblock Springer, 2007.

\bibitem{panteleev_quantum_2021}
Pavel Panteleev and Gleb Kalachev.
\newblock Quantum ldpc codes with almost linear minimum distance.
\newblock {\em IEEE Transactions on Information Theory}, 68(1):213--229, 2021.

\bibitem{pryadko_qdistrnd_2022}
Leonid~P. Pryadko, Vadim~A. Shabashov, and Valerii~K. Kozin.
\newblock Qdistrnd: A gap package for computing the distance of quantum
  error-correcting codes.
\newblock {\em Journal of Open Source Software}, 7(71):4120, Mar 2022.

\bibitem{roffe_biastailored_2023}
Joschka Roffe, Lawrence~Z. Cohen, Armanda~O. Quintavalle, Daryus Chandra, and
  Earl~T. Campbell.
\newblock Bias-tailored quantum ldpc codes.
\newblock {\em Quantum}, 7:1005, May 2023.

\bibitem{hillmann_localized_2025}
Timo Hillmann, Lucas Berent, Armanda~O Quintavalle, Jens Eisert, Robert Wille,
  and Joschka Roffe.
\newblock Localized statistics decoding for quantum low-density parity-check
  codes.
\newblock {\em Nature Communications}, 16(1):8214, 2025.

\bibitem{mian_multivariate_2026}
Feroz~Ahmed Mian, Owen Gwilliam, and Stefan Krastanov.
\newblock Multivariate multicycle codes for complete single-shot decoding.
\newblock {\em arXiv preprint arXiv:2601.18879}, 2026.

\bibitem{eisenbud_commutative_2013}
David Eisenbud.
\newblock {\em Commutative algebra: with a view toward algebraic geometry}.
\newblock Springer Science \& Business Media, 2013.

\bibitem{satherwagstaff_homological_nodate}
Keri Sather-Wagstaff.
\newblock Homological algebra book.
\newblock Accessed: 20 May 2026.

\bibitem{benson_representations_1998}
David~J Benson.
\newblock {\em Representations and cohomology: Volume 1, basic representation
  theory of finite groups and associative algebras}, volume~1.
\newblock Cambridge university press, 1998.

\end{thebibliography}

\clearpage
\onecolumngrid
\appendix

\section{Group algebras}
\label{app:group_algebras}
In this appendix we review some known facts about group algebras and their local blocks, providing a more technical account of the results stated in \Cref{sec:2BGACs-block-decomposition}.
We refer the reader to Refs.~\cite{atiyah_introduction_2018, weibel_introduction_1994} for a more detailed treatment.

\begin{theorem}[\Cref{sec:2BGACs-block-decomposition}: Block decomposition of a binary Abelian group algebra] \label{thm:block_algebra}
    Let $G$ be a finite Abelian group and let $R = \Ft[G]$.
    Write $G = P \times H$, where $P$ is the Sylow $2$-subgroup of $G$ and $|H|$ is odd.
    Then
    \begin{align*}
        R \simeq \kk_1[P] \times \cdots \times \kk_\nu[P]
    \end{align*}
    where $\kk_i \simeq \mathbb{F}_{2^{s_i}}$ is a finite field extension of $\Ft$.
    If $P = \{\mathrm{id}\}$, every block $R_i = \kk_i[P]$ is a field.
    Otherwise, every block is an Artinian local ring $(R_i, \mm_i)$, which is not a field.
\end{theorem}
The decomposition follows from Maschke's theorem applied to the odd part $\Ft[H]$, together with
\begin{align*}
    \Ft[G] \simeq \Ft[H]\otimes_{\Ft} \Ft[P].
\end{align*}
Inside $R$, the decomposition is given by the \emph{primitive central idempotents}: elements $e_i \in R$ with $e_i^2=e_i$, $e_i e_j =0$ for $i\neq j$ and
\begin{align*}
    e_1 + \dots +e_\nu = 1.
\end{align*}
Since the set of primitive central idempotents is unique for every finite commutative algebra, the blocks of $R$ are uniquely identified as
\begin{align*}
    R_i = e_i R.
\end{align*}
\begin{theorem}[Peirce decomposition of modules and homomorphisms]
    The block decomposition of $R$ passes to modules and homomorphisms: 
    \begin{enumerate}[(i)]
        \item Every finite free $R$-module splits as 
        \begin{align*}
            M \simeq \bigoplus_i^\nu e_iM
        \end{align*}
        with $e_iM$ an $R_i$-module.
        \item Every $R$-linear map $f: M \rightarrow N$ is block diagonal: $f(e_i M) \subseteq e_i N$, and $f$ can be written as
        \begin{align*}
            f =\bigoplus_{i = 1}^{\nu} f_i, \quad f_i = f\mid_{e_i M}.
        \end{align*}
        \item For $e_i \neq e_j$, $\operatorname{Hom}_R(e_i M, e_jN) =0$ and
        \begin{align*}
            \operatorname{Hom}_{R}(M, N) \simeq \prod_{i = 1}^{\nu} \operatorname{Hom}_{R_i}(e_iM, e_i N).
        \end{align*}
    \end{enumerate}
\end{theorem}
Applied to a code complex $C_\bullet$ over $R$, these two theorems decompose the complex and every $R$-linear chain map into independent block components.
The same decomposition then descends to homology.
\begin{corollary}[\Cref{sec:2BGACs-block-decomposition}: Block decomposition of homology] \label{cor:decomposition_homology}
    The block decomposition of a group algebra passes to homology for finite free $R$-complexes.
\end{corollary}

\begin{proof}
    Since every differential is $R$-linear,
    \begin{align*}
        \partial_n(e_ix) = e_i \partial_n(x),
    \end{align*}
    the subcomplex $e_i C_\bullet$ is well-defined:
    \begin{align*}
        C_{n, (i)} \xrightarrow[]{\partial_n^{(i)}} C_{n-1, (i)} \, ,
    \end{align*}
    and the orthogonal idempotents decompose the kernels and images of every differential.
    Hence
    \begin{align*}
        H_n(C_\bullet) \simeq \prod_i H_n(e_i C_\bullet).
    \end{align*}
\end{proof}
\begin{proposition} \label{prop:extend_endo_module}
    Let $R_i$ be a block of a finite group algebra over a field, let $M \subseteq R_i^n$ be a submodule and $\phi: M \rightarrow R_i^m$ be $R_i$-linear.
    Then $\phi$ extends to an $R_i$-linear map $\tilde \phi: R_i^n \rightarrow R_i^m$.
\end{proposition}
\begin{proof}
    A finite group algebra over a field is a symmetric Frobenius algebra and hence self-injective.
    A direct factor $R_i$ is therefore self-injective, so every finite free $R_i$-module is injective.
    The defining extension property of the injective module $R_i^m$ extends $\phi$ across the inclusion $M \hookrightarrow R_i^n$.
\end{proof}
\begin{lemma}[Nakayama] \label{lemma:nakayama}
    Let $(R, \mm)$ be an Artinian local ring with residue field $\kk = R/\mm$ and let $M$ be a finitely generated $R$-module.
    Then, elements $x_1, \dots, x_t \in M$ generate $M$ as an $R$-module if and only if their residue classes $[x_i] \in M/\mm M$ span $M/\mm M$ as a $\kk$-vector space.
    Equivalently, if $M = \mm M$, then $M = 0$.
\end{lemma}
\begin{definition} \label{def:min_number_of_gen}
    Let $(R, \mm)$ be a local Artinian ring with residue field $\kk = R/\mm$, and let $M$ be a finitely generated $R$-module.
    The \emph{minimal number of generators} of $M$ as an $R$-module is:
    \begin{align*}
        \mu_{R}(M) \coloneqq \dim_{\kk}(M/\mm M).
    \end{align*}
\end{definition}
\begin{definition}
    Let $N$ be a finitely generated module over a local Artinian ring $(R, \mm)$ with residue field $\kk_{\mm} \coloneqq R/\mm$, where $\mm$ is nilpotent.
    The \emph{radical filtration} is the descending chain
    \begin{align*}
        N \supseteq \mm N \supseteq \mm^2 N \supseteq \dots \supseteq 0.
    \end{align*}
    The quotients
    \begin{align*}
        \mm^{t-1}N/\mm^t N
    \end{align*}
    are $\mathbb{K}_{\mm}$-vector spaces, called the \emph{Loewy layers}.
    The Loewy length $\ell(N)$ is
    \begin{align*}
        \ell(N) \coloneqq \min\{t \ge 0 \mid \mm^tN =0\},
    \end{align*}
    and it is always finite as $\mm$ is nilpotent. 
    Last, the \emph{width} $\mu(N)$
    \begin{align*}
        \mu(N) = \dim_{\mathbb{K}_{\mm}}(N/\mm N)
    \end{align*}
    is, by Nakayama's lemma, the minimal number of generators of $N$ as an $R$-module.
\end{definition}
\begin{definition}
     A ring is a \emph{chain ring} if its ideals are totally ordered by inclusion.
     For a local Artinian ring $(R_i, \mm_i)$ this is equivalent to $\mm_i$ being principal.
     In such case, the ideals are the powers of $\mm_i$, up to the Loewy length of $R_i$.
     A block $\Ft[P] \otimes \kk_i = \kk_i[P]$, where $P$ is a $2$-group, is a chain ring if and only if $P$ is cyclic, in which case $R_i \simeq \kk_i[x]/(x^{2^k}-1)$.
\end{definition} 
\section{Codes as Koszul complexes} \label{app:koszul}
Some commonly studied group-algebra codes, such as generalised bicycle and bivariate bicycle codes, have associated chain complexes that are Koszul complexes, see e.g., Refs.~\cite{eberhardt_logical_2024,breuckmann_logical_2025,mian_multivariate_2026}.
This structure leads to a particularly clean description of scalar surgery, which we develop here.

The key observation is that the mapping cone of scalar surgery on a Koszul complex of length $n-1$ is itself a Koszul complex of length $n$.
Consequently, the merged code obtained from the surgery is canonically associated with a code of the same type, now of length $n$ rather than $n-1$.
In this section, we collect a few facts about Koszul complexes; for a general treatment of Koszul homology, we refer the interested reader to Ref.~\cite{eisenbud_commutative_2013}.
\begin{definition}[Koszul complex] \label{def:koszul}
    Let $R = \Ft[G]$ be a commutative finite algebra and $c_1, \dots, c_n\in R$. The \emph{Koszul complex} $K_\bullet(c_1, \dots, c_n)$ is inductively defined as the length-$n$ chain complex:
    \begin{align*}
        &n = 1  \quad K_{\bullet}(c_1) : \quad &0 \rightarrow R \xrightarrow{c_1} R \rightarrow 0 \\
        &n \ge 2  \quad K_{\bullet}(c_1, \dots, c_{n}) : \quad &\cone(\mu_{c_{n}}^{(c_1, \dots, c_{n-1})}),
    \end{align*}
    where the chain map $\mu_{c_{n}}^{(c_1, \dots, c_{n-1})}$, indicates the scalar multiplication by $c_n$ over  $K_{\bullet}(c_1, \dots, c_{n-1})$.
    
    We indicate by $H_q^{(n)}$ the $q$-th homology group of $K_\bullet(c_1, \dots, c_n)$, where the sequence $c_1, \dots, c_n$ is omitted:
    \begin{align*}
        H_q^{(n)} \coloneqq H_q(K_\bullet(c_1, \dots, c_n)).
    \end{align*}
\end{definition}
From the definition, on $K_\bullet(c_1)$, we have:
\begin{align*}
    H_0^{(1)} &=  R/c_1R, \\
    H_1^{(1)} & = (0:_R c_1).
\end{align*}
\begin{lemma} \label{lemma:dim_h_i}
    Let $R$ be a commutative ring with unit, and let $c_1, \dots, c_n \in R$. If $K_\bullet(c_1, \dots, c_n)$ denotes the standard Koszul complex on the sequence $c_1, \dots, c_n$, then for every $i \ge 0$:
    \begin{align*}
        K_i(c_1, \dots, c_{n}) \simeq R^{\binom{n}{i}}.
    \end{align*}
\end{lemma}
\begin{proof}
    The claim follows directly from the standard definition of Koszul complex as the exterior algebra complex:
    \begin{align*}
        K_i(c_1, \dots, c_n) = \bigwedge^i R^n,
    \end{align*}
    together with the fact that $\bigwedge^i R^n$ is free of rank $\binom{n}{i}$, see~\cite{eisenbud_commutative_2013}.
\end{proof}
\begin{proposition}[Koszul homology] \label{prop:homology_koszul}
    Let $K_\bullet(c_1, \dots, c_{n-1})$ be a Koszul complex, and consider the chain map
    $\mu_{c_n}$ given at each degree by multiplication by $c_n \in R$.
    
    For any $q \in \mathbb{Z}$, there is a short exact sequence:
    \begin{align} \label{eq:ses_koszul}
        0 \longrightarrow \nicefrac{H_q^{(n-1)}}{c_n H_q^{(n-1)}}
        \longrightarrow H_q^{(n)} \longrightarrow \left(0:_{H_{q-1}^{(n-1)}} c_n\right) \longrightarrow 0,
    \end{align}
    where:
    \begin{align*}
        \left(0:_{H_{q-1}^{(n-1)}} c_n\right) = \left\{x \in H_{q-1}^{(n-1)} \mid c_n x = 0 \right\}.
    \end{align*}
\end{proposition}
\begin{proof}
    See Ref.~\cite{satherwagstaff_homological_nodate} and Ref.~\cite{ide_faulttolerant_2025}.
    Let $f_\bullet = \mu_{c_n}^{(c_1, \dots, c_{n-1})}$ so that by \Cref{def:koszul},
    \begin{align*}
        K_{\bullet}(c_1, \dots, c_n) = \cone(f_\bullet)
    \end{align*}
    and therefore
    \begin{align}
        H_{q}^{(n-1)}(\cone(f_\bullet)) = H_{q}^{(n)}\, .
            \label{eq:cone_deg_up}
    \end{align}
    The short exact sequence of the mapping cone~\cite{ide_faulttolerant_2025,cowtan_fast_2025,baspin_fast_2025},
    \begin{align*}
        0 \rightarrow K_\bullet(c_1, \dots, c_{n-1})\xrightarrow[]{\iota} \cone(f_\bullet) \xrightarrow[]{\pi}K_\bullet(c_1, \dots, c_{n-1})[-1] \rightarrow 0
    \end{align*}
    induces,  using the equality \cref{eq:cone_deg_up} at $q+1$ and $q$, the long exact sequence in  homology 
    \begin{align*}
        \cdots \xrightarrow{\iota_*} H_{q+1}^{(n)}\xrightarrow{\pi_*}H_q^{(n-1)} \xrightarrow[]{f_*}H_{q}^{(n-1)}\xrightarrow[]{\iota_*} H_{q}^{(n)} \xrightarrow{\pi_*} \cdots.
    \end{align*}
    The connecting map is:
    \begin{align*}
        f_*([x]) = [c_n x] = c_n [x]
    \end{align*}
    and therefore $f_*$ is multiplication by $c_n$ on the homology.
    
    Combining this with exactness, the long exact sequence yields
    \begin{align*}
    \ker \iota_* &= c_n H_{q}^{(n-1)}\\
    \im \iota_* &= \ker \pi_*\\
    \im \pi_* &= \left( 0:_{H_{q-1}^{(n-1)}} c_n\right)
    \end{align*}
    and therefore~\cref{eq:ses_koszul} follows.
\end{proof}
Note that~\Cref{prop:homology_koszul} says that $H_{q}^{(n)}$ contains $H_{q}^{(n-1)}/c_n H_{q}^{(n-1)}$ as a submodule, and the quotient of $H_{q}^{(n)}$ by this submodule is isomorphic to $(0:_{H_{q-1}^{(n-1)}} c_n)$.
In general the sequence does not split, nevertheless the homology of a Koszul complex has a precise structure in terms of its defining sequence $c_1, \dots, c_n$.
\begin{theorem}[Block structure of Koszul homology] \label{thm:block_koszul_homology}
    Let $(R_i, \mm_i)$ be the local blocks of the commutative group algebra $R$ and $\pi_i :R \rightarrow R_i$ the natural projections.
    We explicitly indicate by $K_{\bullet}(c_1, \dots, c_n; R)$ the Koszul complex  $K_{\bullet}(c_1, \dots, c_n; R)$ over $R$, when specifying the ambient space, here $R$, is necessary.
    The complex $K_{\bullet}(c_1, \dots, c_n; R)$ splits blockwise 
    \begin{align*}
        K_\bullet(c_1, \dots, c_n; R) \simeq \prod_i K_{\bullet}(c_1^{(i)}, \dots, c_n^{(i)}; R_i ), \quad c_j^{(i)} = \pi_i(c_j),
    \end{align*}
    and, for any $q \in \mathbb{Z}$, its homology splits as well
    \begin{align*}
        H_{q}^{(n)} \simeq \prod_i H_q(c_1^{(i)}, \dots, c_n^{(i)}; R_i).
    \end{align*}
    Moreover, given the ideal $J =(c_1, \dots, c_n) \subseteq R$, for every $0 \le q \le n$, the $q$-th homology group is non-trivial on the block $R_i$ if and only if the $\pi_i$ projection of $J$ is contained in the maximal ideal $\mm_i$:
    \begin{align*}
        H_q(c_1^{(i)}, \dots, c_n^{(i)}; R_i) \neq 0 \quad\Longleftrightarrow \quad \pi_i(J) \subseteq \mm_i .
    \end{align*}
    In particular the homology on a block either vanishes at every degree, or is nonzero at every degree.
\end{theorem}
\begin{proof}
    By~\Cref{cor:decomposition_homology}, the homology splits blockwise.
    We here prove the structure of the homology on each block. 
    To ease the notation, we drop the index $i$ and let $(S, \mm)$ be a local Artinian ring, $c_1, \dots, c_n \in S$ and $J = (c_1, \dots, c_n) \subseteq S$, so that we consider the Koszul complex $K_\bullet(c_1, \dots, c_n; S)$.
    We prove that
    \begin{align*}
         H_{q}^{(n)} \neq 0 \Longleftrightarrow J \subseteq \mm. 
    \end{align*}
    As $S$ is local Artinian, every element either is nilpotent and belongs to $\mm$ or is a unit and belongs to $S\setminus \mm$, see Ref.~\cite{atiyah_introduction_2018}.
    
    $(\Longrightarrow)$. If $J \not \subseteq \mm$, then $H_q^{(n)} = 0$ for all $q$.
    
    If $J \not \subseteq \mm$, then at least one $c_j$ is a unit.
    
    We proceed via induction on $n$, the length of the Koszul sequence $c_1, \dots, c_n$.
    
    $[n = 1]$. Multiplication $\mu_{c_1}$ is bijective if $c_1$ is a unit. Hence, by~\Cref{def:koszul}:
    \begin{align*}
        H_0^{(1)} = S/c_1 S = 0\\
        H_1^{(1)} = (0:_S c_1) = 0
    \end{align*}
    as desired.
    
    $[n\ge 2]$. Assume that $c_n$ is a unit, and consider the short exact sequence~\cref{eq:ses_koszul}.
    Similarly to the $n=1$ case, the outer terms are zero:
    \begin{align*}
        H_q^{(n-1)}/c_n H_q^{(n-1)} = 0 \quad \text{and} \quad \left(0:_{H_{q-1}^{(n-1)}} c_n\right) = 0
    \end{align*}
    and therefore $H_q^{(n)} = 0$ for all $q$.
    Assume instead that $c_n \in \mm$ but at least one of $c_1, \dots, c_{n-1}$ is unit.
    By the induction hypothesis
    \begin{align*}
        H_{q}^{(n-1)} = 0
    \end{align*}
    for all $q$ and again both outer terms of~\cref{eq:ses_koszul} are zero for all $q$.
    
    $(\Longleftarrow)$.
    If $J \subseteq \mm$, then $H_{q}^{(n)} \neq 0$ for all $q$.
    
    If $J \subseteq \mm$, then every $c_i$ is nilpotent and therefore, for every non-zero $S$-module $M$:
    \begin{align*}
        M/c_iM \neq 0 \quad \text{and} \quad (0:_M c_i) \neq 0.
    \end{align*}
    Induction step on $n$.
    
    $[n=1]$.
    By \Cref{def:koszul}, both $H_0^{(1)}\neq 0$ and $H_{1}^{(1)} \neq 0$.
    
    $[n\ge 2]$.
    By the induction hypothesis $H_{q}^{(n-1)} \neq 0$ for all $q$.
    In particular, the term $H_{q}^{(n-1)}/c_nH_{q}^{(n-1)} \neq 0$ if $c_n$ is nilpotent and via~\cref{eq:ses_koszul}, $H_{q}^{(n)} \neq 0$ for all $0 \le q \le n-1$.
    For $q = n$, we have $\left( 0:_{H_{n-1}^{(n-1)}} c_n\right) \neq 0$ and again via~\cref{eq:ses_koszul}, $H_{n}^{(n)} \neq 0$.
\end{proof}
In the notation of the dictionary in \Cref{sec:2BGACs-block-decomposition}, the support condition of \Cref{thm:block_koszul_homology}, reads
\begin{align*}
    H_{q, (i)} \neq 0 \Longleftrightarrow J \subseteq P_i \Longleftrightarrow Q_i =B/(T_i + J) \neq 0.
    \label{eq:j_condition}
\end{align*}
We use the last equivalent condition in \Cref{alg:koszul_poly}, where we compute $\dim_{\Ft} Q_i$ directly from a Gr\"obner basis of $T_i + J$.

\subsection{Scalar surgery on Koszul complexes}\label{sec:scalar-surgery-koszul}
Length-$2$ and length-$3$ Koszul complexes are of particular interest for quantum coding theory and surgery applications, so we here focus on these cases.

Explicitly, for $n=2$, elements $a, b \in R$ define the commutative diagram
\begin{center}
    \begin{tikzcd}[sep=2.25em]
        0 & R & R & 0 \\
        0 & R & R & 0
        \arrow[from=1-1, to=1-2]
        \arrow["a", from=1-2, to=1-3]
        \arrow[from=1-3, to=1-4]
        \arrow[from=2-1, to=2-2]
        \arrow["a"', from=2-2, to=2-3]
        \arrow[from=2-3, to=2-4]
        \arrow["b"', from=1-2, to=2-2]
        \arrow["b", from=1-3, to=2-3]
    \end{tikzcd}
\end{center}
which yields the cone complex:
\begin{align*}
      K_{\bullet}(a, b) : R\xrightarrow[]{\begin{pmatrix}
          a \\ b
      \end{pmatrix}}R^2 \xrightarrow[]{\begin{pmatrix}
          b & a
      \end{pmatrix}} R
\end{align*}
as expected. Given $c \in R$, scalar surgery over $K_\bullet(a,b)$ reads
\begin{equation} \label{eq:3d_koszul}
    \begin{tikzcd}[column sep=3.5em, row sep=2em]
        R \arrow[r, "{\binom{a}{b}}"]
          \arrow[ddr, "c" near start, dashed]
        & R^2 \arrow[r, "{(b,\,a)}"]
          \arrow[ddr, "{c \cdot \idty}" near start, dashed]
        & R \arrow[ddr, "c" near start, dashed]
        & {} \\
        \\
        {} & R \arrow[r, "{\binom{a}{b}}"']
        & R^2 \arrow[r, "{(b,\,a)}"']
        & R
    \end{tikzcd}
\end{equation}
where the base code sits on the second row, the auxiliary code sits on the first one and the chain map $c$ is shown by dashed arrows.
By~\Cref{def:koszul}, the mapping cone of $\mu_{c}^{(a, b)}$ is the length-$3$ Koszul complex $K_\bullet(a, b, c)$.
Thus, the merged code produced by surgery is canonically a Koszul-defined code on the sequence $(a, b, c)$: a code of the same family as the base code, where the defining sequence is extended by one entry.
In particular the family of Koszul-defined CSS codes is closed under scalar surgery, with each round extending the defining sequence by one element. 

Working out the cone matrices explicitly, the merged complex $K_\bullet(a, b, c)$ reads:
\begin{align*} \label{eq:3d_koszul_explicit}
    R \xrightarrow{\;\begin{pmatrix} a \\ b \\ c \end{pmatrix}\;}
    R^3 \xrightarrow{\;\begin{pmatrix} b & a & 0 \\ c & 0 & a \\ 0 & c & b \end{pmatrix}\;}
    R^3 \xrightarrow{\;\begin{pmatrix} c & b & a \end{pmatrix}\;}
    R.
\end{align*}
Alternatively, we can unroll the complex in~\cref{eq:3d_koszul} and exhibit its symmetric structure of $K_\bullet(a, b, c)$ with respect to the ring elements $a, b$ and $c$,
\begin{equation*} \label{eq:3d_1}
    \begin{tikzcd}[column sep=1.5em, row sep=1.6em]
    	& \textcolor{black}{R} && \textcolor{black}{R}  \\
    	\textcolor{black}{R}  && \textcolor{black}{R} \\
    	\\
    	& R && R \\
    	R && R
    	\arrow["a", from=1-2, to=1-4]
    	\arrow["b"', from=1-2, to=2-1]
    	\arrow[dotted, from=1-2, to=4-2]
    	\arrow[from=1-4, to=2-3]
    	\arrow["c",dotted, from=1-4, to=4-4]
    	\arrow[from=2-1, to=2-3]
    	\arrow[dotted, from=2-1, to=5-1]
    	\arrow[dotted, from=2-3, to=5-3]
    	\arrow[from=4-2, to=4-4]
    	\arrow[from=4-2, to=5-1]
    	\arrow[from=4-4, to=5-3]
    	\arrow[from=5-1, to=5-3]
    \end{tikzcd}
\end{equation*}
where parallel maps have the same label. 
By \Cref{prop:homology_koszul}, the first homology group of the complex in~\cref{eq:3d_koszul} has the short exact sequence:
\begin{align*}
    0 \longrightarrow H_{1}(a, b)/c H_{1}(a, b) \longrightarrow H_{1}(a, b, c) \longrightarrow (0:_{H_0(a, b)} c) \longrightarrow 0\, .
\end{align*}
After forgetting the $R$-module structure, this short exact sequence splits as a sequence of $\Ft$-vector spaces.
Hence, after adjoining $c$, the new first homology group after adjoining $c$, is built from two terms over $\Ft$:
\begin{enumerate}[(i)]
    \item a quotient of the old homology $H_1(a,b)$, namely the part that survives modulo multiplication by $c$; 
    \item a new contribution coming from elements in the $0$-th homology level of the auxiliary code, i.e., elements in $H_{0}(a,b)$, annihilated by $c$.
\end{enumerate}
More formally, for $K_\bullet(a, b)$, we have the following.
\begin{proposition}[Dimension of length-2 Koszul codes] \label{prop:dim_koszul_two}
    Let $R = \Ft[G]$ for a finite Abelian group $G$, and $a, b \in R$, $J = (a, b) \subseteq R$. 
    For the Koszul complex $K_\bullet(a, b)$ we have:
    \begin{align*}
        H_0 \simeq R/J, \quad H_2 = (0:_R J), \quad \dim_{\Ft} H_1 = 2 \dim_{\Ft} R/J.
    \end{align*}
    And in particular if $\mathfrak{b}_J$ is any standard monomial basis for $(a, b, x_1^{\ell_1}-1, \dots, x_t^{\ell_t}-1)$, then it is also a basis for $H_0 = R/J$. 
\end{proposition}
\begin{proof}
    By definition,
    \begin{align*}
        H_0 &= R/\im \partial_1=R/(a, b)\\
        H_2 &= \ker \partial_2 = (0:_R a) \cap (0:_R b) = (0:_R J).
    \end{align*}
    The split exact sequence~\cref{eq:ses_koszul} over $\Ft$ yields
    \begin{align*}
        \dim_{\Ft} H_1 = \dim_{\Ft} H_0 + \dim_{\Ft} H_2.
    \end{align*}
    We observe that the group algebra $R$ is a symmetric Frobenius algebra~\cite{benson_representations_1998} and therefore $(0:_R J) = \operatorname{Hom}_R(R/J, R)$ so that $\dim_{\Ft} H_2 = \dim_{\Ft} R/J$ and in conclusion, $\dim_{\Ft} H_1 = 2 \dim_{\Ft} R/J$.
    
    Last, it follows from the definition of standard monomials that they are a basis of $R/J$, see~\cite{eisenbud_commutative_2013}.
\end{proof}
In particular, \Cref{prop:dim_koszul_two} gives a formula for the dimension of GB and BB codes, as well as general two-block group algebra codes~\cite{lin_quantum_2024}.
\begin{remark}[Module versus $\Ft$-vector space structure of $H_1$]\Cref{prop:dim_koszul_two} states that $\dim_{\Ft} H_1 = 2 \dim_{\Ft} R/J$, therefore as $\Ft$-vector spaces,
\begin{align*}
    H_1 \simeq_{\Ft} R/J \oplus R/J.
\end{align*}
The isomorphism over $\Ft$ is non canonical, and consequence of the dimension formula. In general this is not an isomorphism of $R/J$-modules, and $H_1$ is not necessarily a free $R/J$-module of rank $2$. 
\label{remark:f2_vs_Q}
\end{remark}
\begin{lemma} \label{lemma:j_zeroes_h1}
    For the Koszul complex $K_\bullet(a, b)$ over $R = \Ft[G]$ with $J = (a, b) \subseteq R$ holds
    \begin{align*}
        J \subseteq (0:_R H_1)
    \end{align*}
    so the action of $c \in R$ on $H_1$ only depends on $[c] \in R/J$.
    In other words, if $c, c' \in R$ generate the same ideal in $R/J$, then the scalar surgeries $\mu_c$ and $\mu_{c'}$ are measurement-equivalent.
\end{lemma}
\begin{proof}
    An element $[(x, y)] \in H_1$  has a representative of the form $h = (x, y) + m(a, b)$ for $(x, y) \in \ker \partial_1$ and $m \in R$. For the generator $a$ of $J$, holds:
    \begin{align*}
        a \cdot h &= (ax, ay) + am(a, b)\\
        &=(ax, bx) + am(a, b), \quad (x, y) \in \ker \partial_1 \Rightarrow bx + ay = 0 \Rightarrow  ay = bx\\
        &= (x + am) (a, b)  \quad \in \im \partial_2\\
        & = 0, \quad \text{ in } H_1.
    \end{align*}
    Since the same holds for the second generator $b$ of $J$,
    \begin{align*}
       J \subseteq (0:_R H_1).
    \end{align*}
    The second statement follows directly from \Cref{def:meas_equiv} of measurement equivalence. 
\end{proof}
What said so far globally for the quotient $Q = R/J$, can be adapted to the local blocks $Q_i = e_i Q$. Let $J_i = e_iJ$, then
\begin{align*}
    Q_i = e_i (R/J) \simeq R_i /J_i.
\end{align*}
\Cref{prop:dim_koszul_two} holds on the local block $Q_i$:
\begin{align*}
    H_{0, (i)} \simeq Q_i, \quad H_{1, (i)} = (0:_{R_i} J_i), \quad \dim_{\Ft} H_{1, (i)} = 2 \dim_{\Ft} Q_i.   
\end{align*}
In particular, we have the following. 
\begin{remark}[Two-block homology is never cyclic] \label{rem:non_cyclic_two_block}
    For a two-block code $K_\bullet(a,b)$, every non-zero homology block $H_{1,(i)}$ is noncyclic as an $R_i$-module.
    Indeed, suppose by contradiction that $H_{1, (i)}$ is cyclic.
    Since $J$ annihilates $H_1$ by \Cref{lemma:j_zeroes_h1}, the module $H_{1, (i)}$ is a quotient of the cyclic module $Q_i = e_i (R/J)$ and so
    \begin{align*}
        \dim_{\Ft} H_{1, (i)} \le \dim_{\Ft} Q_i.
    \end{align*}
    On the other hand, \Cref{prop:dim_koszul_two} gives:
    \begin{align*}
        \dim_{\Ft} H_{1, (i)} = 2 \dim_{\Ft} Q_i.
    \end{align*}
    which is impossible whenever $Q_i \neq 0$, as desired.
    
    Hence \Cref{prop:cyclic_scalar} does not apply directly to GB or BB codes.
    The classification of \Cref{sec:constructions} survives nonetheless: since the scalar action factors through $Q_i$, scalar surgeries are indexed, up to measurement-equivalence, by the cyclic ideals of $Q_i$ — acting identically on both sectors.
\end{remark}

\subsubsection{Classification}
We now show how to efficiently classify scalar surgeries for codes associated to Koszul complexes.
Throughout, we write $J = (a, b)$, call $Q = R/J$ the \emph{coefficient ring} of $K_{\bullet}(a,b)$, and $Q_i = e_i Q \simeq R_i/J_i$ its \emph{local factors}, where $R_i = e_i R$, $J_i = e_i J$.
For $c \in R$ we write $c_i = e_i c$ for its block components and $\bar c_i \equiv c_i \mod J_i$ for their images in the local factors. 

The first homology group of $K_\bullet(a, b)$ is non-cyclic and admits a natural decomposition into two sectors, see~\Cref{rem:non_cyclic_two_block}.
Scalar surgery, by construction, cannot distinguish between these two sectors and necessarily acts on both in the same way. 
Nevertheless, scalar surgery remains a powerful tool in this setting, and can be classified efficiently (\Cref{alg:koszul_poly}) due to the following observations:
\begin{enumerate}[(i)]
    \item The homology of $K_\bullet(a,b)$ is non-zero on a local block $(R_i, \mm_i)$ of $R$ if and only if $\pi_i(J) \subseteq \mm_i$, where $\pi_i$ is the projection $\pi_i : R \rightarrow R_i$, see \Cref{thm:block_koszul_homology}. 
    \item The coefficient ring $Q = R/J$ coincides with the $0$th homology , $Q = H_0$. 
    Each local factor $Q_i= e_i Q$ is a cyclic $R_i$-module, generated by the idempotent $e_i$ modulo the ideal $J_i$.
    The standard monomials of a Gr\"obner basis of $J$ yield an $\Ft$-basis of $Q=R/J$, see \Cref{prop:dim_koszul_two}.
    \item The measured module $c_i H_{1, (i)}$ depends only on
    \begin{align*}
        \bar c_i \equiv c_i \mod J_i.
    \end{align*}
    Thus the scalar surgeries on the $i$th block are parametrised by the cyclic ideals of $Q_i = e_i(R/J)$. 
    Different cyclic ideals may induce the same surgery, but enumerating the cyclic ideals of $Q_i$ always yields a complete set of scalar surgery representatives, possibly with redundancies.
    When $Q_i$ is a chain ring, these are the layers of the radical filtration
    \begin{align*}
        Q_i \supseteq \mm_i Q_i \supseteq \mm_i^2 Q_i \supseteq \dots \supseteq 0.
    \end{align*}
    See \Cref{lemma:j_zeroes_h1}.
    \item Since the idempotents are orthogonal and sum to one, every element is the sum of its components:
    \begin{align*}
        c = \sum_i c_i
    \end{align*}
    and $c_i c_j = 0$ for $i \neq j$. This is why at line \ref{line:sum_c} of \Cref{alg:koszul_poly} the scalar $c$ is a sum: choosing a generator $\bar c_i \in Q_i$ per supported block and lifting each to $c_i \in R$, then $c = \sum_i e_ic_i$ is the desired element. 
    Scalar surgery with $\bar c_i$ measures $\dim_{\Ft}{\bar c_i H_{1, (i)}}$ logical operators. This is equal to $2 \dim_{\Ft}\bar c_i Q_i$ whenever $H_1 \simeq_Q Q \oplus Q$ as $Q$-module, see \Cref{remark:f2_vs_Q}.
\end{enumerate}
\begin{algorithm}[t] \label{alg:koszul_poly}
    \caption{Scalar surgery classification for two-block group algebra codes}
    \begin{algorithmic}[1]
        \Require $B = \Ft[x_1, \dots, x_t]$,
        \Statex $I = (x_1^{\ell_1}{-}1, \dots, x_t^{\ell_t}{-}1)$
        \Statex $R = B/I$,
        \Statex $a, b \in B$, $J = (a, b)$
        
        \Ensure at least one scalar surgery per measurement class
        \State compute the minimal primary decomposition of $I$
        \Statex $I = T_1 \cap \dots \cap T_\nu$, \Statex with associated primes $P_i = \sqrt{T_i}$
        \State compute a Gr\"obner basis of $J + I$ and its standard monomials $\mathfrak{b}_J$ 
        \Statex \Comment{$k = 2|\mathfrak{b}_J|$}
        \State $\mathcal{S} \gets \{i \mid J \subseteq P_i\}$
            \Statex \Comment{reduce $a, b$ modulo $P_i$}
        \State $\mathcal{A} \gets \emptyset$
        \For{each block $i \in \mathcal{S}$}
            \State form $Q_i = B/(T_i + J)$
             \State{ $\mathcal{A} \gets \mathcal{A} \cup \{\text{cyclic ideals of } Q_i\}$ \label{line:cyclic}}
        \EndFor
        \State choose a target class and a local generator $\bar c_i \in Q_i$
        \State lift each $\bar c_i \in Q_i$ to a polynomial $c_i \in B$
        \State $c \gets \sum_{i \in \mathcal{S}} e_ic_i$ \label{line:sum_c}
        \State $\mu_{c'} \gets$ optimise $\mu_c$ via \Cref{alg:optimise}
        \State \Return $\mu_{c'}$
    \end{algorithmic}
\end{algorithm}
At line~\ref{line:cyclic} of~\Cref{alg:koszul_poly} we compute, for each local factor, its cyclic ideals.
We can do that via $\Ft$-linear algebra as follows.
\begin{enumerate}[(i)]
    \item (Find a nice basis of $Q_i$). The standard monomials $\mathfrak{b}_J$ form an $\Ft$-basis of $Q = B/(I+J)$ \cite{cox_ideals_2007}.
    The $\Ft$ representation of a polynomial $p$ is the unique remainder $p_{\mathcal{G}}$ under division of $p$ by the Gr\"obner basis $\mathcal{G}$,
    a linear combination of $\mathfrak{b}_J$. 
    A basis $\mathfrak{b}_i$ of $Q_i$ is obtained by row-reducing the vectors
    $\{(e_i\beta)_\mathcal{G} \mid \beta \in \mathfrak{b}_J\}$.    
    \item (Generate and store cyclic ideals). The cyclic ideal generated by $\bar c \in Q_i$ is
    \begin{align*}
        \bar c\, Q_i =\operatorname{span}_{\Ft}\{(\bar c \,b)_G \mid b \in \mathfrak{b}_i\}.
    \end{align*}
    We store each ideal as a matrix in reduced row echelon form, so equality of ideals is equality of matrices.
    \item (Chain-ring shortcut).
    First compute the chain of subspaces
    \begin{align*}
        Q_i \supseteq \mm_iQ_i \supseteq \mm_i^2Q_i \supseteq \dots
    \supseteq 0 \, ,
    \end{align*}
    where $\mm_i = P_i/T_i$ and $P_i = \sqrt{T_i}$.
    
    The local factor $Q_i$ is a chain ring  if and only if:
    \begin{align*} \dim_{\Ft}\big(\mm_iQ_i/\mm_i^2Q_i\big) \le \dim_{\Ft}\big(Q_i/\mm_iQ_i\big).
    \end{align*}
    In such a case, the subspaces $\mm_i^t Q_i$ are \emph{all} of its ideals and a generator of $\mm_i^tQ_i$ is any of its elements outside in $\mm_i^{t+1}Q_i$.
\end{enumerate}
GB and BB codes as discussed in~\Cref{sec:ls-constructions-gb} and~\Cref{sec:ls-constructions-bb} are, respectively, the univariate and bivariate special cases of two-block group algebra codes.

Thus, in the following subsections we specialise the techniques above (in particular \Cref{alg:koszul_poly}) to a broad variety of code constructions and work through a dedicated example code instance for each.
For an overview, see~\Cref{tab:surgery_summary}.

\subsection{On the choice of the modulus $x^\ell - 1$ of the quotient ring $\Ft[x]/(x^\ell-1)$} \label{app:choice_of_ell}
The roots of $x^\ell-1$ are the $\ell$-th roots of unity over extension fields of $\Ft$.
We assume $\ell$ odd, for $\ell' = 2^s \ell$ with $\ell$ odd, $x^{\ell'} -1 = (x^{\ell}-1)^{2^s}$ and the following argument applies to the odd part.
Over $\Ft[x]$, if $\vartheta$ is a root of $p(x)$, so is $\vartheta^{2}$, and the degree of the minimal polynomial of $\vartheta$ is exactly the size of the orbit of $\vartheta$ under the Frobenius endomorphism $\vartheta \mapsto \vartheta^2$.
As all the $\ell$th root of unity have order that divides $\ell$, to understand how the polynomial $x^{\ell} - 1$ factors into irreducible (i.e., minimal polynomial of its roots on some extension field), we need to look at integers divisors $\lambda$ of $\ell$.

For each divisor $\lambda$ of $\ell$ there are $\phi(\lambda)$ roots of order exactly $\lambda$, where $\phi$ is Euler's totient function.
Each root has minimal polynomial of degree $\mathrm{ord}_\lambda(2)$, the multiplicative order of $2$ modulo $\lambda$.
Hence $x^\ell - 1$ has, for each $\lambda \mid \ell$, exactly $\phi(\lambda)/\mathrm{ord}_\lambda(2)$ irreducible factors of degree $\mathrm{ord}_\lambda(2)$. 
For instance, $\lambda = 1$ corresponds to the root $1$, with minimal polynomial $x-1$, which always divides $x^\ell - 1$; the primitive roots ($\lambda = \ell$) contribute $\phi(\ell)/\mathrm{ord}_\ell(2)$ factors of degree $\mathrm{ord}_\ell(2)$.
 
Picking $\ell = 2^r-1$ as done in Refs.~\cite{webster_explicit_2025, webster_pinnacle_2026} yields a very specific factorisation of the modulus, as any irreducible factor of $x^{2^{r}-1}-1$ has degree that divides $r$.
In fact, $2^r \equiv 1 \mod \ell$, and so for each $\lambda$ divisor of $\ell$, $\mathrm{ord}_\lambda(2)$ divides $r$.
More explicitly, we can use the fact that each element in $\mathbb{F}_{2^r}$ satisfies the equation $x^{2^{r}} = x$, and that
\begin{align*}
    x^{2^{r}}- x = \prod_{\deg(p) \mid r, \, p \text{ irreducible}} p.
\end{align*} 

\section{Deferred proofs} \label{app:main_proofs}
\subsection{Proofs of \Cref{sec:ls-theory}}
Each lifted surgery $f \in \mathrm{LS}_R(C_\bullet)$, induces a $R$-linear map on homology
\begin{align*}
    f^1_*: H_1(C_\bullet) \rightarrow H_1(C_\bullet)\, ,
\end{align*}
and the set of logical operators measured is $\im f_*^1$; 
as $f_*^1$ is $R$-linear, its image is an $R$-submodule of $H_1(C_\bullet)$.
In particular, the addressability of lifted surgery is, in first instance, governed by the structure of the $R$-submodules of the homology module. 

In the following, we fix a block $R_i$.
We write
\begin{align*}
    Z_i = \ker \partial_1^{(i)}, \quad B_i = \im \partial_2^{(i)}, H_{1, (i)} = Z_i/B_i,
\end{align*}
and write $\pi : Z_i \twoheadrightarrow H_{1, (i)}$ for the quotient map.
\begin{proposition}[Compatibility condition for measurable maps] \label{lemma:lifted_endo_fix_boundaries}
       An endomorphism $g \in \operatorname{End}_{R_i}(H_{1, (i)})$ is measurable if and only if there exists an $R_i$-linear map $\phi: Z_i \rightarrow Z_i$ on the cycles that 
       satisfies the compatibility condition
       \begin{align*}
           \phi(B_i) \subseteq B_i, \quad g\pi= \pi\phi.
       \end{align*}
\end{proposition}
\begin{proof}       
    Suppose $g \in \im(\star_i)$ and let $g_\bullet = (g_2, g_1, g_0)$ be its preimage in $\ls_{R_i}(C_{\bullet, (i)})$.
    First, $g_1$ maps cycles to cycles, because $\partial_1 g_1 = g_0 \partial_1$, so it restricts to $Z_i$.
    Second, $g_1$ maps boundaries to boundaries because $g_1\partial_2 = \partial_2 g_2$.
    Thus, letting $\phi = g_1|_{Z_i}$, $\phi(B_i) \subseteq B_i$ and $\pi \phi = g \pi$ since $g = g_*^1$ is by definition the map induced by $f_1$ on the quotient $H_{1, (i)} = Z_i/B_i$.
    Conversely, let $\phi \in \operatorname{End}_{R_i}(Z_i)$ be such that $\phi(B_i) \subseteq B_i$ and $\pi \phi = g \pi$.
    By \Cref{prop:extend_endo_module}, $\phi$ extends to the finite free $R_i$-module $C_{1}^{(i)}$ as $f_1 : C_1^{(i)} \rightarrow C_{1}^{(i)}$.
    By construction
    \begin{align*}
       f_1(Z_i) \subseteq Z_i, \quad f_1(B_i) \subseteq B_i.
    \end{align*}
    Second, we define $f'_0$ as
    \begin{align*}
       f'_0 : \im \partial_1^{(i)} &\rightarrow  \im \partial_1^{(i)}\\
       \partial_1^{(i)}(v) &\longmapsto \partial_1^{(i)}f_1(v) \, .
    \end{align*}
    The map $f_0'$ is well defined: if $\partial_1^{(i)}(v) = \partial_1^{(i)}(w)$ then $v + w \in Z_i$, hence $f_1(v+w) \in Z_i$ and $\partial_1^{(i)}f_1(v) = \partial_1^{(i)}f_1(w)$.
   
    Again using \Cref{prop:extend_endo_module}, we extend $f_0'$ to $f_0 \in \operatorname{End}_{R_i}(C_0^{(i)})$. Last, since $C_2^{(i)}$ is free, it is projective. Hence, the map $f_1\partial_2^{(i)}:C_2^{(i)} \rightarrow B_i$, lifts through the surjection  $\partial_2^{(i)} :C_2^{(i)} \twoheadrightarrow B_i$ . Thus, there exists $f_2$ such that the following diagram commutes:
    \begin{equation*}
    \begin{tikzcd}
        & C_2^{(i)} \\
        C_2^{(i)} & B_i
        \arrow["\partial_2^{(i)}" ,two heads, from=1-2, to=2-2]	
        \arrow["f_1\partial_2^{(i)}"', from=2-1, to=2-2]
        \arrow[dashed, "\exists f_2", from=2-1, to=1-2]
    \end{tikzcd}
    \end{equation*}
   In conclusion $f_\bullet = (f_2, f_1, f_0)$ is a lifted surgery and its image via the homology-action map is $g$, as desired.
\end{proof}
The following proposition characterises the submodules that appear as image of an endomorphism, and so it serves to describe the addressable submodules when we restrict to measurable maps. 
\begin{proposition}[Submodules as quotients] \label{prop:quotient}
    Let $M$ be an $R$-module and $W \subseteq M$ a submodule. There exists $g \in \operatorname{End}_R(M)$ such that $\im g = W$ if and only if $W$ is isomorphic to a quotient of $M$:
    \begin{align*}
        M/N \simeq W
    \end{align*}
    for some submodule $N \subseteq M$.
\end{proposition}
\begin{proof}
    Let $W = \im g$. Then $g: M \twoheadrightarrow W$ is surjective and, by the First Isomorphism Theorem, $W \simeq M/ \ker g$. Vice versa, let $\phi : M/N \xrightarrow[\sim]{} W$ be an isomorphism and $\pi: M \twoheadrightarrow M/N$ the natural projection.
    The map $g = \iota \circ \phi \circ \pi $, where $\iota : W \hookrightarrow M$ is the inclusion, is an endomorphism of $M$ with image $W$.
    In particular:
    \begin{align*}
        \mu_{R}(W)\leq \mu_{R}(M),
    \end{align*}
    where $\mu_R(M)$ is the minimal number of generators of $M$ as $R$ module (see \Cref{def:min_number_of_gen} in \Cref{app:group_algebras}).
\end{proof}

\Cref{lemma:lifted_endo_fix_boundaries} describes which endomorphisms of $H_{1, (i)}$ are induced by lifted surgeries.
\Cref{prop:quotient} describes which submodules, in general, can occur as the image of some endomorphism.
Combining these two propositions, $W$ is addressable if and only if:
\begin{enumerate}
    \item there exists $g\in\operatorname{End}_{R_i}(H_{1,(i)})$
    such that $\im g=W$;
    \item for at least one such $g$, there exists
    $\phi\in\operatorname{End}_{R_i}(Z_i)$ satisfying:
    \begin{align*}
        \phi(B_i)\subseteq B_i,
        \qquad
        \pi\phi=g\pi.
    \end{align*}
\end{enumerate}
As the following theorem shows, these two conditions are also sufficient. 
\addressablemodule*
\phantomsection
\label{proof:addressable_module}
\begin{proof}
    $(1) \Leftrightarrow (2)$. By \Cref{lemma:lifted_endo_fix_boundaries}, the measurable maps are exactly the quotient maps induced by endomorphisms of cycles that preserve boundaries.
    In other words, $g\in \mathcal{M}_i$ if and only if there exists $\phi \in \operatorname{End}_{R_i}(Z_i)$ such that
    \begin{align*}
      \phi(B_i) \subseteq B_i, \quad g\pi |_{Z_i} = \pi  \phi |_{Z_i}.
    \end{align*}
    Hence, $W \in \mathcal{A}_i$ if and only if $W= \im g$ for such a quotient map $g$.
    Since
    \begin{align*}
      \im g = g(\pi(Z_i)) = \pi (\phi(Z_i)),
    \end{align*}
    this proves the equivalence of $(1)$ and $(2)$.
    
    $(2) \Rightarrow (3)$.
    Take $\phi \in \operatorname{End}_{R_i}(Z_i)$ and define
    \begin{align*}
      g: H_{1, (i)} &\longrightarrow H_{1, (i)} \\
      \pi(z) &\longmapsto \pi(\phi(z))
    \end{align*}
    which is well defined because $\pi : Z_i \twoheadrightarrow H_{1, (i)}$ is surjective, and $\phi(B_i) \subseteq B_i$, so $\phi$ passes to the quotient.
    By construction, $\im g =W$, and hence $(3)$ is verified.
    
    $(3) \Rightarrow (2)$.
    Let $b \in B_i$:
    \begin{align*}
        \pi(\phi(b)) = g(\pi(b)) = g(0) = 0
    \end{align*}
    since $\pi(b) = 0$ for any $b \in B_i$ by definition of the quotient map $\pi$.
    In particular, $\phi(b) \in \ker \pi = B_i$, $\forall b \in B_i$, hence $\phi(B_i) \subseteq B_i$. Last, $W = \im g = \pi(\phi(Z_i))$.
    
    $(2) \Rightarrow (4)$.
    Given $\phi \in \operatorname{End}_{R_i}(Z_i)$, let $g$ be the associated endomorphism of $H_{1, (i)}$.
    Set
    \begin{align*}
        N = \ker g.
    \end{align*}
    For $z \in Z_i$,
    \begin{align*}    
        z\in \phi^{-1}(B_i) &\Leftrightarrow\phi(z) \in B_i \\
        &\Leftrightarrow \pi(\phi(z)) = 0 \\
        &\Leftrightarrow g(\pi(z)) = 0 \\
        &\Leftrightarrow \pi(z) \in N \\
        &\Leftrightarrow z \in \pi^{-1}(N).
    \end{align*}
    Hence, $\phi^{-1}(B_i) = \pi^{-1}(N)$ and $(4)$ holds.
    
    $(4) \Rightarrow (2)$.
    Since $0 \in N$,
    \begin{align*}
        B_i = \pi^{-1}(0) \subseteq \pi^{-1}(N) = \phi^{-1}(B_i)
    \end{align*}
    therefore
    \begin{align*}
        \phi(B_i) \subseteq B_i
    \end{align*}
    and together with $W = \pi(\phi(Z_i))$, $(2)$ holds. 
\end{proof}
\Cref{thm:addressability} characterises the scope of lifted surgery and yields an important corollary: whenever the blocks of $R$ are fields, every $R$-linear endomorphisms is measurable and every $R$-submodule addressable.
\surgeryfield*
\phantomsection
\label{proof:surgery_field}
\begin{proof}
    Since $R_i = \kk_i$ is a field, then the short exact sequence
    \begin{align*}
        0 \rightarrow B_i \rightarrow Z_i \rightarrow H_{1, (i)} \rightarrow 0
    \end{align*}
    splits as vector spaces, so
    \begin{align*}
        Z_i \simeq B_i \oplus H_{1, (i)}.
    \end{align*}
    Since on a field a linear subspace always has a complement:
    \begin{align*}
        C_{1}^{(i)} \simeq Z_i \oplus T_i ,
    \end{align*}
    and combining these two splittings, with some minor misuse of notation, we can write:
    \begin{align*}
        C_{1}^{(i)} \simeq B_i \oplus H_{1, (i)} \oplus T_i.
    \end{align*}
    Now let $g \in \operatorname{End}_{\kk_i}(H_{1, (i)})$.
    Define $f_1 \in \operatorname{End}_{\kk_i}(C_1^{(i)})$:
    \begin{align*}
        f_1: B_i \oplus H_{1, (i)} \oplus T_i &\longrightarrow C_{1}^{(i)}\\
        (b, h, t) &\longmapsto g(h)
    \end{align*}
    and set $f_\bullet = (0, f_1, 0)$. It is easy to verify that $f_\bullet$ is a chain map. Moreover, $f_*^1 = g$ by construction and therefore every endomorphism of $H_{1, (i)}$ is measurable.
    \begin{align*}
        \im(\star_i)= \operatorname{End}_{\kk_i}(H_{1, (i)}),
    \end{align*}
    Last, if $W \subseteq H_{1, (i)}$ is a $\kk_i$-linear subspace, by taking its complement, we can write 
    \begin{align*}
        H_{1, (i)} \simeq W \oplus U
    \end{align*}
    so that $W \simeq \im \pi_W$ where $\pi_W \in \operatorname{End}_{\kk_i}(H_{1, (i)})$ is the natural projection onto $W$.
    Since $\pi_W$ is measurable, $W$ is addressable, as desired.
\end{proof}
\surgerycyclic*
\phantomsection
\label{proof:cyclic_scalar}
\begin{proof}
    Since $H_{1, (i)}$ is cyclic, there exist a generator $h \in H_{1, (i)}$ such that $H_{1, (i)} = \{rh \mid r \in R_i\}$.
    The map
    \begin{align*}
       \phi: R_i &\rightarrow H_{1, (i)} \\
       r &\longmapsto rh
    \end{align*}
    is surjective with $\ker \phi = (0:_{R_i} h) = (0:_{R_i} H_{1, (i)})$.
    By the First Isomorphism Theorem
    \begin{align*}
       H_{1, (i)} \simeq R_i/(0:_{R_i} h).
    \end{align*}
    Every $R_i$-linear endomorphism of the cyclic module $H_{1, (i)}$ is uniquely determined by the image of the generator $h$ and hence, by linearity, coincides with multiplication by some $r \in R_i$.
    In other words, $\operatorname{End}_{R_i}(H_{1, (i)}) \simeq R_i/(0:_{R_i} h)$ and each endomorphism is scalar multiplication.
    By \Cref{lemma:scalar_multiplication} every scalar multiplication is measurable, hence $\mathcal{M}_i = \operatorname{End}_{R_i}(H_{1, (i)})$.
    For the second claim, first observe that a finite chain ring is principal, $\mm_i  = (m)$, and every ideal is a power $(m^t) = \mm_i^t$ for some integer $t$.
    In particular
    \begin{align*}
        (0:_{R_i} h) = \mm_i^{\ell}
    \end{align*}
    where $\ell = \ell(H_{1, (i)})$ is the Loewy length of $H_{1, (i)}$ (see \cref{app:group_algebras}).
    The submodules of $H_{1, (i)}\simeq R_i/(0:_{R_i} h) = R_i/\mm_i^{\ell}$ are the ideals of $R_i$ that contain $\mm_i^{\ell}$, hence they are each image of the scalar surgery $\mu_{m^t}$:
    \begin{align*}
       \mm_i^t H_{1, (i)} = m^t H_{1, (i)} = \im (\mu_{m^t})_*, \quad 0 \le t \le \ell.
    \end{align*}
    In particular, each submodule is addressable via the scalar surgery $\mu_{m^t}$ and the chain
    \begin{align} \label{eq:strict_chain}
       \mm_i^{t}H_{1, (i)} \supseteq \mm_i^{t+1} H_{1, (i)}
    \end{align}
    is strict, as by Nakayama's lemma $\mm_i^{t}H_{1, (i)} = \mm_i^{t+1} H_{1, (i)}$ implies $\mm_i^{t}H_{1, (i)} = 0$.
    
    Last, suppose $\mu_c$ and $\mu_{c'}$ are measurement-equivalent i.e.\, $cH_{1, (i)} = c' H_{1, (i)}$.
    If $c H_{1, (i)} = 0$, then
    \begin{align*}
       c+ c' \in (0:_{R_i} H_{1, (i)})
    \end{align*}
    we have an annihilator shift.
    Otherwise, write $c = u \cdot m^{\rho(c)}$ and $c' = u' \cdot m^{\rho(c')}$, where $u, u'$ are units.
    Then
    \begin{align*}
       c H_{1, (i)} = m^{\rho(c)} H_{1, (i)} = c' H_{1, (i)} = m^{\rho(c')} H_{1, (i)}
    \end{align*}
    as the chain \cref{eq:strict_chain} is strict, $\rho(c) = \rho(c')$ and therefore $c' = u' u^{-1} c$ is a unit rescaling. 
\end{proof}
\homotopicequivalence*
\begin{proof}
    Let $h$ be an homotopy for $f$ and $f'$.
    If $\partial z = 0$, then $(f-f')z = \partial h(z) + h \partial(z) = \partial(hz) \in \im \partial$ is a boundary and so $[fz] = [f'z]$ and therefore homotopic maps induces the same map on homology. 
    
    Second, we exhibit a chain isomorphism between $\cone(f)$ and $\cone(f')$.
    Consider the maps $\psi_{h, n}:  \cone(f) \longrightarrow \cone(f')$
    \begin{align*}
        \psi_{h, n}: D_{n-1}\oplus C_n &\longrightarrow D_{n-1} \oplus C_n \\
        (x, y) &\longmapsto (x, h_{n-1}(x)+y)
    \end{align*}
    To show that $\psi_{h}$ is a chain isomorphism, we need to verify that the squares commute:
    \begin{align*}
        \psi_{h,n-1}\partial_n^{\cone(f)} = \partial_n^{\cone(f')}\psi_{h, n}
    \end{align*}
    which holds if and only if:
    \begin{align*}
        \begin{pmatrix}
            \partial_{n-1} & 0\\
            h_{n-2} \partial_{n-1} + f_{n-1} & \partial_n
        \end{pmatrix} =  \begin{pmatrix}
            \partial_{n-1} & 0\\
            \partial_{n-1}  h_{n-1} + f'_{n-1} & \partial_n
        \end{pmatrix}
    \end{align*}
    which holds whenever $f\simeq f'$ via the homotopy $h$.
    Thus, $\psi_h$ is a well-defined chain map from $\cone(f)$ to $\cone(f')$.
    To show that it is actually an isomorphism, we observe that it is invertible:
    \begin{align*}
        \psi_h^{-1}(x, y) = (x, y+h(x)) = \psi_h(x, y).
    \end{align*}
    In conclusion:
    \begin{align*}
        \psi_h \partial^{\cone(f)} = \partial^{\cone(f')}\psi_h,
    \end{align*}
    and  $\cone(f) \simeq \cone(f')$ are isomorphic as chain complexes.
\end{proof}
\subsection{Proofs of~\Cref{sec:bridging}}\label{sec:app-proofs-bridging}
\propOneBridging*
\begin{proof}
This claim essentially follows from the SES and rank-nullity, see also~\cref{eq:cone-ses}, and~\cite{cowtan_fast_2025,ide_faulttolerant_2025}.

Recall that from the SES $0 \to C \oplus C' \xrightarrow{\iota} M_\bullet \xrightarrow{\pi}  D[-1] \to 0$ with $\iota(x) = (x,0), \pi(x,y) = y$, we obtain a LES in homology which around deg 1 reads 
\begin{align*}
    H_1(D_\bullet) \xrightarrow{H_1(m)} H_1(C\oplus C') \xrightarrow{\iota_*} H_1(M_\bullet) \\ \xrightarrow{\pi_*} H_0(D_\bullet) \xrightarrow{H_0(m)} H_0(C \oplus C').
\end{align*}
We check the kernels and images of the maps in the LES: 
Exactness at $H_1(C \oplus C')$ implies $\ker \iota_* = \im m_*$.
Moreover, we have $\ker \pi_* = \im \iota_*, \im \pi_* = \ker H_0(m_\bullet)$ and applying rank-nullity to $\iota_*, \pi_*,$ and $H_0(m_\bullet)$ we obtain
\begin{align*}
    & \dim_{\Ft} H_1(M_\bullet)=  \dim_{\Ft} \ker \pi_* + \dim_{\Ft} \im \pi_* = \\ \nonumber
    &  \dim_{\Ft} \im \iota_* + \dim_{\Ft} \ker H_0(m_\bullet) = \\ \nonumber
    & (\dim_{\Ft} H_1(C \oplus C') - \rank_{\Ft} H_1(m_\bullet)) + \\\nonumber
    & (\dim_{\Ft} H_0(D_\bullet) - \rank_{\Ft} (H_0 (m_\bullet)) \, . \qedhere
\end{align*} 
\end{proof}
\propTwoBridging*
\begin{proof}
By \Cref{prop:bridging_measured_set}, the measured logical space is
\begin{align*}
   \im m_* =
    \left\{ 
        \bigl(f_*[d],g_*[d]\bigr) \;\middle|\; [d]\in H_1(D_\bullet) 
    \right\}. 
\end{align*}
Suppose first that the kernels differ \(K_f\neq K_g\). After interchanging \(f_*\) and \(g_*\), if necessary, there exists
\begin{equation*}
    [d]\in K_g\setminus K_f.
\end{equation*}
It follows that $m_*[d]=\bigl(f_*[d],0\bigr)\neq 0$, so the measured space contains a nonzero logical operator supported
entirely on the first base code.

Conversely, suppose that \(K_f=K_g=:K\). 
If $m_*[d]=\bigl(f_*[d],0\bigr)$, then \(g_*[d]=0\), so \([d]\in K_g=K_f\) and hence $f_*[d]=0$ and \(m_*[d]=0\). 
Interchanging \(f_*\) and \(g_*\) gives the same conclusion for a measured operator supported entirely on the second base code. 
Thus every nonzero measured logical operator is measured jointly.

Furthermore,
\begin{align*}
    \ker m_* &= \left\{[d]\mid f_*[d]=0,\ g_*[d]=0\right\}\\
             &= K_f\cap K_g =K.
\end{align*}
The first isomorphism theorem therefore gives $\im m_*\simeq H_1(D_\bullet)/K$, and rank--nullity over $\Ft$ gives
\begin{align*}
    \rank_{\Ft}m_* = \dim_{\Ft}H_1(D_\bullet)-\dim_{\Ft}K = r \, .  
\end{align*}

Finally, because the classes
\begin{align*}
    [d_1]+K,\dots,[d_r]+K    
\end{align*}
form a basis of \(H_1(D_\bullet)/K\), their images under the induced isomorphism form a basis of $\im m_*$. 
Since \(f_\bullet\) and \(g_\bullet\) are chain maps,
\begin{align*}
    \partial_1^C f_1(d_j)=f_0\partial_1^D(d_j)=0,
        \quad
    \partial_1^{C'}g_1(d_j)=g_0\partial_1^D(d_j)=0 \, ,    
\end{align*}
so \(f_1(d_j)\) and \(g_1(d_j)\) are cycles representing the stated logical operators. \qedhere
\end{proof}

\section{Description of measurable maps via $\operatorname{Ext}$} \label{app:ext}
Fix a block $R_i$.
By definition of chain complex, we have the short exact sequence of $R_i$-modules:
\begin{align*}
    \xi_i : 0 \rightarrow B_i \rightarrow Z_i \rightarrow H_{1, (i)} \rightarrow 0,
\end{align*}
and we denote by the same symbol its extension class $\xi_i$
\begin{align*}
    \xi_i \in \operatorname{Ext}^1_{R_i}(H_{1, (i)}, B_i).
\end{align*}
$\operatorname{Ext}^1$ is a bifunctor: an endomorphism $\beta$ of $B_i$ acts by pushout $\xi_i \mapsto \beta_* \xi_i$ and an endomorphism $g$ of $H_{1, (i)}$ acts by pullback $\xi_i \mapsto g^* \xi_i$.
Then the image of the homology-action map is:
\begin{align*}
    \im (\star_i)= \left\{
    g \in \operatorname{End}_{R_i}(H_{1, (i)})
    \;\middle|\;
    \begin{aligned}
        &\beta_* \xi_i = g^*\xi_i,\\
        &\text{for some } \beta \in \operatorname{End}_{R_i}(B_i)
    \end{aligned}
    \right\}.
\end{align*}
Thus:
\begin{enumerate}
    \item On a field block $\kk_i$, $\operatorname{Ext}^1_{R_i}(H_{1, (i)}, B_i)=0$ so all $\kk_i$-linear endomorphisms are measurable
    \begin{align*}
        \mathcal{M}_i = \operatorname{End}_{\kk_i}(H_{1, (i)}).
    \end{align*}
    \item On a local non-field block $(R_i, \mm_i)$ the class $\xi_i$ is an obstruction. However if the short exact sequence $\xi_i$ splits
    \begin{align*}
        \mathcal{M}_i = \operatorname{End}_{R_i}(H_{1, (i)})
    \end{align*}
    even if $R_i$ is not a field.
    Furthermore, scalar maps are always measurable as they always satisfy the condition. 
\end{enumerate}

\section{Additional codes}

\subsection{$\llbracket 882,48,16 \rrbracket$ code}\label{sec:ls-constructions-ghp}
Consider the $\llbracket 882, 48, 16\rrbracket$ code taken from Example B2 of Ref.~\cite{panteleev_degenerate_2021}.
There the code is defined as generalised hypergraph product code over 
\begin{align*}
    R = \Ft[x]/(x^{63}-1)
\end{align*}
by the $7 \times 7$ circulant matrix
\begin{align*}
    A = \begin{pmatrix}
        x^{27} &0 &0 & 1 & x^{18} & x^{27} & 1 \\
        1& x^{27} & 0&0 & 1 & x^{18} & x^{27} \\
        x^{27} & 1& x^{27} & 0&0 & 1 & x^{18} \\
        x^{18} & x^{27} & 1& x^{27} & 0& 0& 1 \\
        1& x^{18} & x^{27} & 1& x^{27} &0 &0 \\
        0&1& x^{18} & x^{27} & 1& x^{27} & 0 \\
        0&0&1& x^{18} & x^{27} & 1& x^{27}  \\
    \end{pmatrix},
\end{align*}
and the polynomial $b = (1 + x + x^6)$. 
Encoding the circulant shift matrix into the variable $y$, and adding the relation $y^7 - 1$, the circulant matrix $A$ can be written as:
\begin{align*}
    a(x, y)= x^{27} + y^{3} + x^{18}y^4 + x^{27}y^5  + y^6,
\end{align*}
and the code is described by the Koszul complex $K_\bullet(a, b)$ over 
\begin{align*}
    R=B/I, \quad B = \Ft[x, y], \quad I = (x^{63}-1, y^7-1)\, .
\end{align*}
As we have already seen several examples of scalar surgery analysis in the main text, we leave the details to the interested reader.

\subsection{$\llbracket 198, 8, 16 \rrbracket$ radial code}\label{sec:app-r311-code}
A quantum radial code~\cite{scruby_highthreshold_2026} defined over $R=\Ft[x]/(x^{11}-1)$ with protographs
\begin{align*}
    A =
    \begin{pmatrix}
    x^6 & x^4 & x^9 \\
    x^2 & x^5 & x^3 \\
    x^2 & x^9 & x^9
    \end{pmatrix},
    \qquad
    B =
    \begin{pmatrix}
    x^5    & 1   & x^4 \\
    1      & x^3 & x^9 \\
    x^{10} & x^4 & x^7
    \end{pmatrix}\, .    
\end{align*}
Like the $\llbracket 90, 8, 10 \rrbracket$ radial code, the homology is entirely supported on one completely fine block.

\end{document}